\documentclass[12pt,a4paper]{article}

\usepackage{ifpdf}
\usepackage{color}
\usepackage{latexsym}                   
\usepackage{epsfig}                     
\usepackage{graphicx}
\usepackage{amssymb, amsmath, amsthm}   
\usepackage{enumerate}                  
\usepackage{mathrsfs}                   
\usepackage{geometry}                                                                          
\usepackage{verbatim}                                                                           

\usepackage[round]{natbib}
\usepackage{setspace}
\usepackage[nolists]{endfloat}

\newtheorem{thm}{Theorem}[]
\newtheorem{lem}[thm]{Lemma}

\newtheorem{prop}[thm]{Proposition}

\theoremstyle{definition}

\theoremstyle{definition}

\theoremstyle{definition}

\newcommand	{\RR}{\mathbb{R}}
\newcommand	{\PP}{\mathbb{P}}
\newcommand	{\EE}{\mathbb{E}}

\newcommand {\timestop}{T}

\newenvironment{myindentpar}[1]%
     {\begin{list}{}%
             {\setlength{\leftmargin}{#1}}%
             \item[]%
     }
     {\end{list}}

\newcommand{\tanjaR}{\textcolor{black}}

\begin{document}

\title{
Macro-evolutionary models and coalescent point processes:\\
\medskip
The shape and probability of reconstructed phylogenies
}
\author{\textsc{By Amaury Lambert$^{1,2}$ \& Tanja Stadler$^{3,*}$}
}
\date{}
\maketitle
\noindent\textsc{$^1$
UPMC Univ Paris 06\\
Laboratoire de Probabilités et Modèles Aléatoires CNRS UMR 7599}\\
\noindent\textsc{$^2$
Collège de France\\
Center for Interdisciplinary Research in Biology CNRS UMR 7241\\
Paris, France}\\
\textsc{E-mail: }amaury.lambert@upmc.fr\\
\textsc{URL: }http://www.proba.jussieu.fr/pageperso/amaury/index.htm\\
\\
\noindent\textsc{$^3$ Institute of Integrative Biology\\
ETH Z\"{u}rich\\
Universit\"{a}tsstrasse 16\\
8092 Z\"{u}rich\\
Switzerland}\\
\textsc{E-mail: }tanja.stadler@env.ethz.ch\\
\textsc{URL: }http://www.tb.ethz.ch/people/tstadler\\
\\
\noindent{$^*$ corresponding author}\\

\newpage

\doublespacing

\section*{\textsc{Abstract}}

Forward-time models of diversification (i.e., speciation and extinction) produce phylogenetic trees that grow ``vertically'' as time goes by. Pruning the extinct lineages out of such trees leads to natural models for reconstructed trees (i.e., phylogenies of extant species). Alternatively, reconstructed trees can be modelled by coalescent point processes (CPP), where trees grow ``horizontally'' by the sequential addition of vertical edges. Each new edge starts at some random speciation time and ends at the present time; speciation times are drawn from the same distribution independently. CPP lead to extremely fast computation of tree likelihoods and simulation of reconstructed trees. Their topology always follows the uniform distribution on ranked tree shapes (URT).

We characterize which forward-time models lead to URT reconstructed trees and among these, which lead to CPP reconstructed trees. We show that for any ``asymmetric'' diversification model  in which speciation rates only depend on time and extinction rates only depend on time and on a non-heritable trait (e.g., age), the reconstructed tree is CPP, even if extant species are incompletely sampled. If rates additionally depend on the number of species, the reconstructed tree is (only) URT (but not CPP). We characterize the common distribution of speciation times in the CPP description, and discuss incomplete species sampling as well as three special model cases in detail: 1) extinction rate does not depend on a trait; 2) rates do not depend on time; 3) mass extinctions may happen additionally at certain points in the past.

\bigskip
\noindent
\textit{Running head.} Macro-evolutionary models and coalescent point processes.\\
\textit{Key words and phrases.}  random tree; birth-death process; incomplete sampling; likelihood; inference.


\newpage

\section*{\textsc{Introduction}}


A general, lineage-based forward-time model of macroevolution assumes that speciation and extinction rates may change as a function of (i) time,  (ii) number of co-existing species, (iii) a non-heritable trait (i.e., a trait changing in the same way in all species independently), 
and  (iv) a heritable trait 
  \citep{Stadler2011PNAScommentrary}. Here speciation is implicitly assumed asymmetric, i.e., we distinguish between mother and daughter species and the trait of the mother species is assumed to remain unchanged upon speciation.

Phylogenetic trees of only extant species, i.e., {\it reconstructed phylogenies}, contain information about past speciation and extinction dynamics. \citep{Thompson1975,Nee1994} provides analytic equations for calculating the likelihood of a reconstructed phylogenetic tree, under a model of diversification assuming constant speciation rate $\lambda$ and constant extinction rate $\mu$. These equations allow maximum likelihood inference of speciation and extinction rates based on the knowledge of the reconstructed phylogeny.  In \citet{AlPo2005}, it is  actually shown that when $\lambda=\mu$, the reconstructed tree viewed from a given stem age $\timestop$ is a {\it coalescent point process} (CPP). In \citet{Yang2006,Gernhard2008JTB}, it was shown that this property holds for any values of $\lambda$ and $\mu$, 
and in \citet{Lambert2010}, this result has been generalized to models where the extinction rate is a function of the age.

A CPP with stem age $\timestop$ is specified by draws from independent and identically distributed (iid) one-dimensional random variables until a value larger than $\timestop$ is drawn. Let the first draw being bigger than $\timestop$ be the $n-$th draw. The first $n-1$ draws $H_1,\ldots,H_{n-1}$  give rise to a phylogenetic tree on $n$ tips in the following way (see Fig. 1b): the $n$ tips of the tree are located in a 2-dimensional plot on  $(0,T),(1,T),\ldots,(n-1,T)$. Now the realization $h_k$ of the random variable $H_k$ in the CPP ($k=1,\ldots,n-1$) is located at $(k,T-h_k)$. We obtain the phylogeny as follows. The first branch in the phylogeny is the line between $(0,T)$ and $(0,0)$. Now we proceed iteratively for $k=1,\ldots, n-1$. We add to the phylogeny the vertical line joining $(k,T)$ to $(k,T-h_k)$. Now we add the horizontal line joining $(k,T-h_k)$ to $(m,T-h_k)$ where $m<k$ is the rightmost pre-existing edge at height $T-h_k$ (dotted lines in Fig. 1b).

Thus we can produce reconstructed trees under the CPP by sequentially sampling `points', i.e. speciation times from left to right (horizontally), compared to simulating sequentially speciation (and extinction) events under a speciation-extinction model from bottom to top (vertically).


The common probability density $f$ of the iid random variables $H_1, H_2, \ldots$ is called the {\it  coalescent density}. Knowing the  coalescent density 
allows us to calculate the likelihood of a given reconstructed tree with stem age $T$ and node depths $h_1, \ldots, h_{n-1}$: it is simply the product of $f(h_i)$ over $i=1,\ldots,n-1$, 
times the probability $r(T)$ that the $n$-th draw is larger than $T$ (i.e., $r(T)=\int_{T}^\infty f(t)dt$). This likelihood function can then be used directly for macroevolutionary parameter inference using maximum likelihood or Bayesian methods. The likelihood of a given phylogeny with stem age $T$ \emph{conditioned} on the  number $n$ of species can be calculated by taking the product of $f_T(h_i)$ over $i=1,\ldots,n-1$, where $f_T$ is $f$ conditioned on the draw being smaller than $T$, that is, $f_T (h_i)= f(h_i)/\left(1-r(T)\right)$. Note that  these calculations can only be done with the knowledge of the stem age or alternatively the crown age (or under some prior distribution thereof). This is in contrast with random tree models that are stationary in time, as those used in population genetics (e.g., Kingman coalescent \citep{Kin82}).

Furthermore, based on the CPP representation,  fast simulation algorithms for phylogenetic trees with stem age $T$  and additionally fixing or not fixing the number of species $n$ were developed: essentially only $n-1$ one-dimensional random variables (corresponding to the $(n-1)$ speciation times) have to be sampled \citep{Stadler2010SystBiol}. As classical forward-in-time simulation tools have to account for each speciation and extinction event, and the number of events may be much higher than $n$ in the case of high extinction rates, the CPP-based simulations are in particular advantageous for high extinction rates.
Even for low extinction rates, the CPP-based method remains much more efficient for the simulation of reconstructed trees when there is a need to fix the number of tips $n$ together with the stem age $T$.

A common feature of the coalescent point process with $n$ tips is that it induces, ignoring time and {\it orientation} (see below), the same distribution on \emph{ranked tree shapes} with $n$ tips as that induced (ignoring labels) by the uniform distribution on ranked trees shapes with $n$ \emph{labelled} tips (URT). Ranked trees are  reconstructed phylogenies in which branch lengths are ignored, but the order of branching times is acknowledged.  The URT distribution is often also called Yule-Harding distribution \cite{Yule1924,Harding1971}.

The aim of this paper is to identify which macroevolutionary models, depending on the scenarios (i)-(iv), give rise to URT reconstructed trees, and among these, which give rise to CPP reconstructed trees. For the latter models, we characterize the  coalescent density of node depths, so that those models can be used for parameter inference and fast simulations. 

We show that  whenever the speciation and/or extinction rates 
depend on (i) time and (ii) number of species, and  the extinction rate possibly further depends on (iii) a non-heritable \emph{asymmetric} (see below) trait, then the distribution induced on ranked tree shapes by the reconstructed tree is URT. We show that if the rates are additionally independent of species number, then the reconstructed tree is a CPP. We provide counterexamples of the last two assertions when the corresponding requirements are not fulfilled, therefore providing a complete characterization of forward-in-time scenarios leading to CPP or to URT reconstructed trees (see Table 1).

We start with a rigorous definition of macroevolutionary models and of associated notions (Section``Macroevolutionary Models''), then define coalescent point processes (Section ``Coalescent point processes''). In the main part of the paper (Section ``Main results''), we characterize the macroevolutionary models which induce URT reconstructed trees, and within these, which induce CPP reconstructed trees. Concerning the latter models, we provide various ways of characterizing the coalescent density of the associated CPP in terms of the model ingredients. We then discuss examples and applications of macroevolutionary models with CPP reconstructed trees (Section ``Three special cases''). We finally study the link between CPP and another popular model for random binary trees with edge lengths, namely the  Kingman coalescent \citep{Kin82} (Section ``What about Kingman coalescent?'').

\section*{\textsc{Macroevolutionary models}}

\subsection*{\textit{Useful definitions}}

We   define a general, lineage-based macroevolutionary model of speciation and extinction. The process starts with one species at time $0$ in the past. A species speciates with rate $\lambda$ and goes extinct with rate $\mu$.
Both rates may change as a function of:
\begin{itemize}
\item[(i)] time;
\item[(ii)] number of co-existing species;
\item[(iii)] a non-heritable trait, i.e., a trait changing in the same way in all species independently, either deterministically like age, or randomly, provided the initial value of the trait follows the same distribution for all species (this distribution may possibly depend on time);
\item[(iv)] a heritable trait, i.e., the initial value of the trait is correlated with the trait value of the mother species at speciation  \citep{Stadler2011PNAScommentrary}. 
\end{itemize}

For models with trait-dependent speciation, conditionally given the initial trait values (at speciation), traits of different species evolve independently through time, with the same probability transitions.  The initial value of a trait is  drawn independently, from the same distribution for all species (non-heritable trait), or from a distribution which depends on mother species trait (heritable trait).

When heritability is less than 100 \%, it becomes important to distinguish between symmetric and asymmetric speciation. Under symmetric speciation, both daughter species are ``new'' species and inherit the mother trait only partially (or not at all in the case of non-heritability). Under asymmetric speciation, one daughter species is the ``new'' species and inherits only partially (and possibly not at all) the mother trait while the other descendant corresponds to the mother species, inheriting the trait to 100 \%.

A non-heritable trait is typically the age of a species. It is equivalent to say that the extinction rate depends on the age and to say that the species lifetime has a probability density which is arbitrary (and not necessarily exponential, as in the case of a constant extinction rate). Another example of a non-heritable trait is speciation stage, as in the model of protracted speciation \citep{rosindell2010protracted, etienne2012prolonging, LME}, where species are incipient following speciation and eventually become good (and can only be detected as species when they are good).

Trees are given an \emph{orientation} \citep{Stadler2009SystBiol}, by distinguishing, upon speciation, between the \emph{left} species (mother species in the case of asymmetric speciation) and the \emph{right} species (daughter species in the case of asymmetric speciation).

The process is stopped at the present (time $\timestop$) leading to a complete tree (Figure \ref{fig:examplePhylo}a; left species is the species with the straight line, right species the added species). The resulting tree consists of extant and extinct species. Pruning all extinct species yields an \emph{ultrametric} tree with stem age $\timestop$ (Figure \ref{fig:examplePhylo}b), in the sense that all tip points are at the same distance $\timestop$ from the root point, called the {\it reconstructed tree} \citep{Nee1994}. Note that when pruning lineages, the orientation of each new branch is obtained by the orientation of the most ancestral branch in the complete tree corresponding to the new branch in the reconstructed tree. 

Most available phylogenies are not complete, in the sense that not all extant species descending from the same ancestor species are sampled and included in the phylogeny. 
There are four main ways considered in the literature for randomly removing tips from a phylogenetic tree: the $p-$sampling or Bernoulli model \citep{Stadler2009JTB,  lambert2009allelic,Stadler2011PNAS,morlon2010inferring,Morlon2011,hallinan2012generalized}, the $n$-sampling model \citep{Stadler2009JTB,etienne2012diversity},  the diversified sampling model \citep{Stadler2011MBE-Sampling}, and the higher-level phylogeny model \citep{Paradis2003,Stadler2012groups}. In the $p$-sampling scheme, given the phylogenetic tree (or the reconstructed tree), each tip is removed independently with probability $1-p$, where $p$ is the so-called \emph{sampling probability}. In the $n$-sampling scheme, given a phylogenetic tree (or a reconstructed tree) with more than $n$ tips, $n$ tips are selected uniformly (e.g., sequentially) and all other tips are removed.  In the diversified sampling scheme, the $n-1$ oldest speciation events are preserved, and each of the $n$ monophyletic clades existing after the $n-1$th speciation event is collapsed into a single lineage. Finally higher-level phylogenies are phylogenies in which monophyletic species clades are collapsed into one tip, and the number of species represented by this tip is recorded. Such phylogenies are common if only one species per (say) genus is added to the phylogeny, but the sizes of the different genera are known. We restrict our  higher-level phylogenies to trees obtained by collapsing descendant monophyletic clade of each lineage existing at a specified time $x_{cut}$ in the past (also called strict higher-level phylogenies \citep{Stadler2012groups}).

An oriented, ultrametric tree with $n$ tips is characterized by its node depths, $h_0=\timestop$ and  $h_1, \ldots, h_{n-1}$, as in Figure \ref{fig:examplePhylo}b. The orientation of the tree implies that $h_i$ ($1\le 1\le n-1$) is the coalescence time between species $i-1$ and species $i$, where species are labelled $0,\ldots, n-1$ from  left to right, and also that  $\max \{h_{i+1},\ldots, h_j\}$ is the coalescence time between species $i$ and species $j$.

\begin{figure}[!ht]
\input{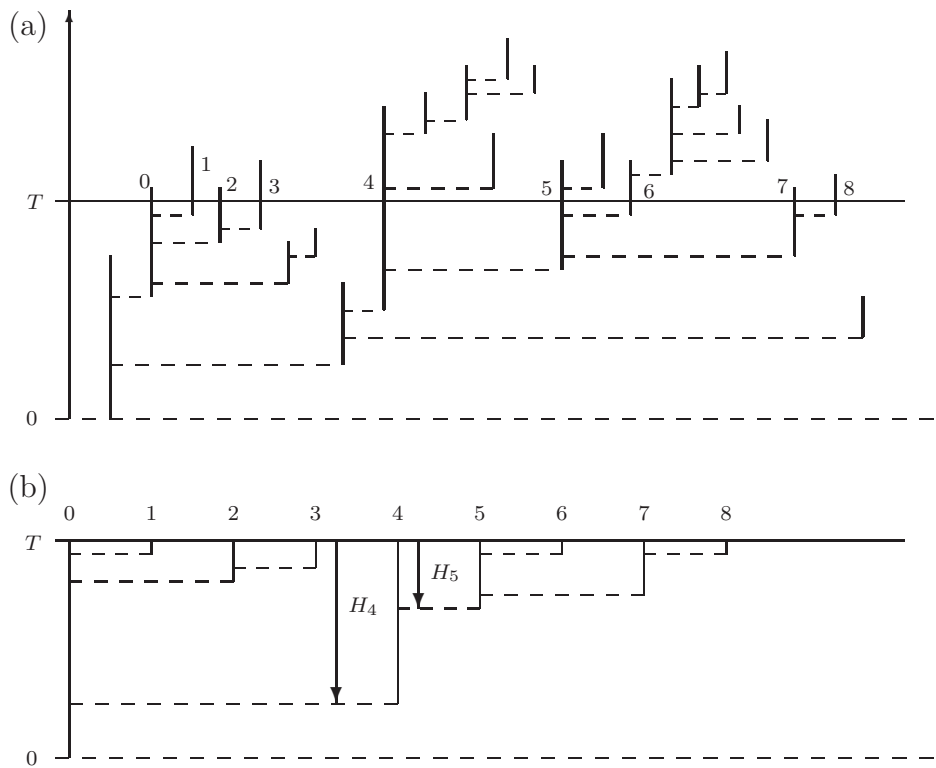}
\caption{a) An oriented phylogenetic  tree generated by a general macroevolutionary lineage-based speciation and extinction model; the $N_T=9$ species extant at $\timestop$ are labelled $0,1, \ldots ,8$ from left to right; b) 
The reconstructed tree obtained from the complete tree in a), showing the coalescence times $H_4$, between species 3 and 4, and $H_5$ between species 4 and 5.}
\label{fig:examplePhylo}
\end{figure}

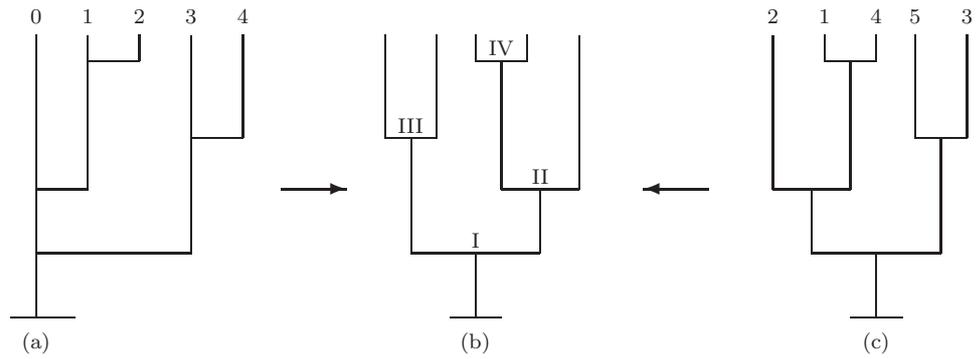
\begin{figure}[ht]
\unitlength 1.7mm 
\linethickness{0.4pt}
\ifx\plotpoint\undefined\newsavebox{\plotpoint}\fi 
\begin{picture}(86,26.5)(-8,0)
\put(4,3){\line(1,0){5}}
\put(6,3){\line(0,1){22}}
\put(10,25){\line(0,-1){12}}
\put(14,25){\line(0,-1){2}}
\put(18,25){\line(0,-1){17}}
\put(18,8){\line(-1,0){12}}
\put(14,23){\line(-1,0){4}}
\put(10,13){\line(-1,0){4}}
\put(22,25){\line(0,-1){8}}
\put(22,17){\line(-1,0){4}}
\put(6,26.5){\makebox(0,0)[cc]{\scriptsize $0$}}
\put(10,26.5){\makebox(0,0)[cc]{\scriptsize $1$}}
\put(14,26.5){\makebox(0,0)[cc]{\scriptsize $2$}}
\put(18,26.5){\makebox(0,0)[cc]{\scriptsize $3$}}
\put(22,26.5){\makebox(0,0)[cc]{\scriptsize $4$}}
\thicklines
\put(25,13){\vector(1,0){5}}
\thinlines
\put(35,17){\line(0,-1){9}}
\put(76,17){\line(0,-1){9}}
\put(33,25){\line(0,-1){8}}
\put(74,25){\line(0,-1){8}}
\put(33,17){\line(1,0){4}}
\put(74,17){\line(1,0){4}}
\put(37,17){\line(0,1){8}}
\put(78,17){\line(0,1){8}}
\put(40,25){\line(0,-1){2}}
\put(67,25){\line(0,-1){2}}
\put(40,23){\line(1,0){4}}
\put(67,23){\line(1,0){4}}
\put(44,23){\line(0,1){2}}
\put(71,23){\line(0,1){2}}
\put(42,23){\line(0,-1){10}}
\put(69,23){\line(0,-1){10}}
\put(42,13){\line(1,0){6}}
\put(63,13){\line(1,0){6}}
\put(48,13){\line(0,1){12}}
\put(63,13){\line(0,1){12}}
\put(45,13){\line(0,-1){5}}
\put(66,13){\line(0,-1){5}}
\put(45,8){\line(-1,0){10}}
\put(76,8){\line(-1,0){10}}
\put(40,8){\line(0,-1){5}}
\put(71,8){\line(0,-1){5}}
\put(38,3){\line(1,0){4}}
\put(69,3){\line(1,0){4}}
\thicklines
\put(58,13){\vector(-1,0){5}}
\put(63,26.5){\makebox(0,0)[cc]{\scriptsize $2$}}
\put(67,26.5){\makebox(0,0)[cc]{\scriptsize $1$}}
\put(71,26.5){\makebox(0,0)[cc]{\scriptsize $4$}}
\put(74,26.5){\makebox(0,0)[cc]{\scriptsize $5$}}
\put(78,26.5){\makebox(0,0)[cc]{\scriptsize $3$}}
\put(6,1){\makebox(0,0)[cc]{\scriptsize (a)}}
\put(40,1){\makebox(0,0)[cc]{\scriptsize (b)}}
\put(71,1){\makebox(0,0)[cc]{\scriptsize (c)}}
\put(40,9){\makebox(0,0)[cc]{\scriptsize I}}
\put(45,14){\makebox(0,0)[cc]{\scriptsize II}}
\put(35,18){\makebox(0,0)[cc]{\scriptsize III}}
\put(42,24){\makebox(0,0)[cc]{\scriptsize IV}}
\end{picture}
\caption{(a) a ranked \emph{oriented} tree with 5 tips (labelled from left to right); (c) a ranked \emph{labelled} tree with 5 tips; (b) the ranked tree shape associated to (a) by ignoring orientation and to (b) by ignoring labels. Under the uniform distribution on ranked oriented trees, the probability of the tree in (a) is $1/(n-1)!= 1/24$; under the uniform distribution on ranked labelled trees, the probability of the tree in (c) is  $2^{n-1}/n! (n-1)! = 1/180$. Under URT, the probability of the tree $\tau$ in (b) is $2^{n-1-c(\tau)}/(n-1)!=1/6$.  The node rankings are indicated by Roman figures, specifying in particular that the split of the three-tip subtree comes before the split of the cherry.
}
\label{fig:exampleURT}
\end{figure}

A {\it ranked tree} is obtained from an ultrametric tree by ignoring branch lengths in the tree but maintaining the information about the order of speciation events. The rank of the most ancestral speciation event is 1, the next speciation event has rank 2, etc (indicated by Roman figures in Figure \ref{fig:exampleURT}b). 
The number of ranked oriented trees on $n$ tips is $(n-1)!$ (number of permutations of the $n-1$ edges different from the leftmost one, which has length equal to $T$). The number of ranked labelled trees with $n$  tips is $n!(n-1)!/2^{n-1}$ (\citet{Edwards1970} and Proposition 2.3.4 in \citet{Steel2003}). 

It can be proven straightforward that the  following two probability distributions on ranked tree shapes with $n$ tips (ranked speciation events, but no orientation, no labels) are equal. These are the probabilities respectively induced
\begin{itemize}
\item
by the uniform distribution on ranked oriented trees after ignoring the orientation;
\item 
by the uniform distribution on ranked labelled trees after ignoring the labels.
\end{itemize}
We denote this probability by URT. By standard calculations, it can be seen that under URT, the probability of a ranked tree shape $\tau$ is
$$\frac{2^{n-1-c(\tau)}}{(n-1)!},$$  where $c(\tau)$ is the number of cherries of $\tau$ (i.e. the number of nodes subtending two tips), see Figure \ref{fig:exampleURT}.


Recall that oriented reconstructed trees are CPP if the node depths are $n$ iid draws with the $n-$th draw being the first draw bigger than $T$. In particular, this means that each reconstructed tree on $n$ tips can be uniquely represented by $n-1$ points, and $n-1$ points uniquely define a reconstructed tree on $n$ tips \citep{Gernhard2008JTB}.
A noticeable feature of the CPP with $n$ tips is that, after ignoring its edge lengths (and orientation), it follows the URT distribution on ranked trees \citep{AlPo2005,Gernhard2008JTB}. Thus,  macroevolutionary model which do not give rise to URT reconstructed trees cannot give rise to CPP reconstructed trees.

We highlight here that any model inducing URT on completely sampled trees also induces URT on trees with incomplete sampling modelled as $p-$ or $n-$ sampling (\citet{Stadler2009SystBiol}, follows from Proposition A5), as well as with diversified sampling (\citet{Stadler2013JEB} and again shown below). Furthermore any model inducing CPP trees induce URT on higher-level phylogenies (as shown below). 

In the following, we will characterize which of the rate dependencies (i)-(iv) induce a URT distribution on complete trees and thus may have a CPP representation. We further investigate which sampling schemes preserve a CPP representation.

\subsection*{\textit{Characterization of macroevolutionary models}}

Recall that trait heritability can be symmetric or asymmetric. For a non-heritable trait, either the traits of both incipient species are reset upon speciation (symmetric speciation, two daughters), or the trait of one (the mother) species remains unchanged upon speciation (asymmetric speciation, one mother and one daughter). We start by showing that there exist symmetric speciation models giving rise to non-URT reconstructed trees,  in simple cases where only one of the speciation/extinction rates is age-dependent and the other rate is constant. 

There are two distinct (ranked or not) trees with 4 tips, the perfectly balanced tree $B$ and the caterpillar tree $C$. Let us start with a symmetric speciation model where extinction rate is zero and speciation rate is $\lambda$ in a small interval $[1-\varepsilon,1]$, and zero outside. If $\lambda$ is sufficiently large, species speciate with high probability at an age close to 1. For $\timestop=2$, reconstructed trees with 4 tips are of type $B$ with arbitrarily high probability, and so cannot follow the URT distribution (see Figure \ref{fig:exampleSym}a).

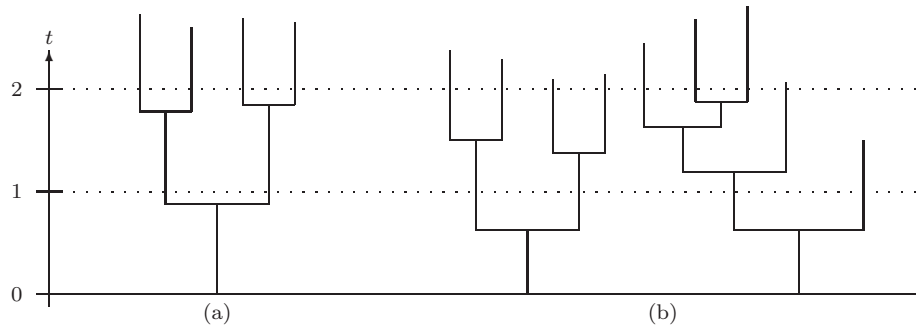
\begin{figure}[ht]
\unitlength 1.7mm 
\linethickness{0.4pt}
\ifx\plotpoint\undefined\newsavebox{\plotpoint}\fi 
\begin{picture}(72,25.375)(-8,0)
\put(5,2){\vector(0,1){20}}
\put(4,11){\line(1,0){2}}
\put(4,19){\line(1,0){2}}
\put(2.5,3){\makebox(0,0)[cc]{\scriptsize $0$}}
\put(2.5,11){\makebox(0,0)[cc]{\scriptsize $1$}}
\put(2.5,19){\makebox(0,0)[cc]{\scriptsize $2$}}
\put(5,23){\makebox(0,0)[cc]{\scriptsize $t$}}
\put(18,3){\line(0,1){7}}
\put(14,10){\line(0,1){7.25}}
\put(16,17.25){\line(-1,0){4}}
\put(12,17.25){\line(0,1){7.5}}
\put(16,17.25){\line(0,1){6.5}}
\put(22,10){\line(0,1){7.75}}
\put(20,17.75){\line(1,0){4}}
\put(24,17.75){\line(0,1){6.375}}
\put(20,17.75){\line(0,1){6.75}}
\put(42,3){\line(0,1){5}}
\put(46,8){\line(-1,0){8}}
\put(38,8){\line(0,1){7}}
\put(46,8){\line(0,1){6}}
\put(36,15){\line(0,1){7}}
\put(40,15){\line(0,1){6.25}}
\put(44,14){\line(0,1){5.75}}
\put(48,14){\line(0,1){6.125}}
\put(36,15){\line(1,0){4}}
\put(44,14){\line(1,0){4}}
\put(63,3){\line(0,1){5}}
\put(68,8){\line(-1,0){10}}
\put(68,8){\line(0,1){7}}
\put(58,8){\line(0,1){4.5}}
\put(62,12.5){\line(-1,0){8}}
\put(62,12.5){\line(0,1){7}}
\put(54,12.5){\line(0,1){3.5}}
\put(57,16){\line(-1,0){6}}
\put(57,16){\line(0,1){2}}
\put(51,16){\line(0,1){6.5}}
\put(55,18){\line(1,0){4}}
\put(59,18){\line(0,1){7.375}}
\put(55,18){\line(0,1){6.375}}
\put(14,10){\line(1,0){8}}
\multiput(4.93,10.93)(.985294,0){69}{{\rule{.4pt}{.4pt}}}
\multiput(4.93,18.93)(.985294,0){69}{{\rule{.4pt}{.4pt}}}
\put(4,3){\line(1,0){68}}
\put(18,1.5){\makebox(0,0)[cc]{\scriptsize (a)}}
\put(52.5,1.5){\makebox(0,0)[cc]{\scriptsize (b)}}
\end{picture}
\caption{Trees with 4 tips at time 2 produced by a symmetric speciation model with age-dependent rates, time flowing upwards, where: (a) extinction rate is 0 and speciation rate is high for ages close to 1 (0 otherwise); (b) extinction rate is high for ages larger than 1 (0 otherwise) and speciation rate is small and constant. Reconstructed trees with 4 tips will be of type $B$ (balanced) with high probability (panels a and b-left). Panel (b-right) shows that a type $C$ tree (caterpillar) requires at least 4 speciations.   
}
\label{fig:exampleSym}
\end{figure}

Now let us consider a symmetric speciation model where speciation rate is constant equal to $\lambda$ and extinction rate is equal to $\mu$ for ages larger than $1-\varepsilon$ and zero otherwise. If $\mu$ is sufficiently large, species 
have lifetimes smaller than (and close to) 1. Assume $\timestop=2$. If all species have lifetimes smaller than 1, at least 4 speciation events have to occur for a reconstructed tree to be of type $C$, whereas only 3 suffice for type $B$ reconstructed trees. If $\lambda$ is very small, most trees never reach $\timestop$, but (the rare) trees with 4 tips are of type $B$ with arbitrarily high probability, and so do not follow the URT distribution (see Figure \ref{fig:exampleSym}b). 

A consequence of the previous paragraph is that symmetric speciation models do not give rise to URT reconstructed trees in general if at least one of the speciation/extinction rates depends on a non-heritable trait. Since non-heritable traits are particular cases of heritable traits, we have shown in general that symmetric trait-dependent speciation models do not give rise to URT reconstructed trees and thus not to CPP. Then we will not consider this class of models further.

From now on, speciation is then assumed to be asymmetric. Recall though that the asymmetry is only important for scenarios (iii) and (iv) (trait-dependent rates). Scenarios (i) (time-dependent rates) and (ii) (rates dependent upon the number of co-existing species) are equivalent under symmetric and asymmetric speciation. We will now investigate which of the scenarios (i) --- (iv) produce URT reconstructed trees, and which of those produce CPP reconstructed trees.


Table 1 summarizes new results obtained in this paper, classifying models and stating which models give rise to URT reconstructed trees, and within these models, which give rise to CPP reconstructed trees. This characterization hopefully facilitates the usage of the different macroevolutionary models in an efficient way.

We provide examples of model classes 4 (speciation rate depends on (iii-iv) a trait, heritable or not; extinction rate can be arbitrary) and 5 (extinction rate depends on (iv) a heritable trait; speciation rate can be arbitrary) leading to non-URT reconstructed trees. This means in particular that such models do not give rise to CPP reconstructed trees in general  (Section ``Model classes 4-5 do not induce URT''). 

We thus focus on speciation rates depending on  (i) time and/or (ii) number of species, and extinction rates not depending on a heritable trait (model classes 1, 2 and 3). In this case, we show that with arbitrary extinction rates (i-iii), we always obtain URT reconstructed trees (Section ``Model classes 1-3 induce URT''). 
However, we show that whenever extinction rates (model class 2) or speciation rates (model class 3) depend on (ii) the number of species, then reconstructed trees are not CPP in general (Section ``Model classes 2-3 do not induce CPP''). 

In fact model classes 1-3 are the only models discussed above belonging to the class of ``species-speciation-exchangeable models'' defined in \cite{Stadler2013JEB} and shown to induce URT; all other models discussed above (including the symmetric speciation models) belong to the ``species-non-exchangeable models''.

\begin{table}[!th]
\begin{center}
\begin{tabular}{|  c | c|c|c|c|c| c|  }
\hline
 Model class & speciation & extinction & URT & CPP \\
  \hline
1&  0,i & 0,i,iii & x & x \\   
 2&  0,i & ii & x & - \\  
 3&  ii & 0-iii & x & - \\ 
  4&iii-iv & 0-iii & - & - \\   
  5 &0-iv & iv & - & -\\
  \hline
\end{tabular}
\end{center}
\caption{Asymmetric speciation models: (0) constant, (i) time-, (ii) number of species-, (iii) non-heritable trait-, (iv) heritable trait- dependent speciation or extinction.
 URT = uniform ranked tree distribution, CPP = coalescent point process. A property (URT, CPP) is satisfied (x) by a model class if it is satisfied under all possible combinations of assumptions on speciation and extinction rates; a property is not satisfied (-) if there is a counterexample for any combination of assumptions.} 
\label{Table:}
\end{table}

\subsection*{\textit{Macroevolutionary model class 1}}

Based on the previous observations, only models in class 1 may give rise to  CPP reconstructed trees. We show that these models actually always give rise to CPP reconstructed trees (Section ``Model class 1 induces CPP'', Theorem \ref{ThmMain}), and so in particular to URT reconstructed trees.
In other words, when the extinction rate $\mu(t,x)$ only depends on (i) time $t$ and (iii) a non-heritable trait $x$, and the speciation rate $\lambda(t)$ only depends on (i) time $t$, the resulting reconstructed trees can be represented by a CPP. 

In Theorem \ref{ThmMain}, we provide a characterization of the one-dimensional coalescent density $f(t)$ of the CPP for any model belonging to class 1. In Proposition \ref{prop:NEW}, we present a way to evaluate, at least numerically, this density. By formulas \eqref{eqn:likelihood1}, \eqref{eqn:likelihood1bis} or \eqref{eqn:likelihood2}, this density can then be used to calculate the likelihood of a reconstructed phylogeny, meaning we can obtain speciation and extinction rate estimates, using maximum likelihood or Bayesian methods, under very general macroevolutionary models.  Further, we are able to simulate reconstructed phylogenies under model class 1 very fast by sampling $n-1$ one-dimensional random variables from the coalescent density conditioned on $T$.

We furthermore show that even under certain incomplete sampling schemes, reconstructed trees of model class 1 remain coalescent point processes (see Section ``Missing tips'').
In Section ``Three special cases'' we discuss some special cases of our model class 1 in detail:
\begin{itemize}
\item Speciation and extinction rates are both functions of time, but uniquely of time ($\lambda(t)$, $\mu(t,x)=\mu(t)$). The coalescent density is given in Proposition \ref{prop:markov}.
The likelihood of the whole tree had previously been derived \citep{Nee1994,Morlon2011, Hoehna} for general $\lambda(t)$, $\mu(t)$,
and in \citep{Stadler2011PNAS}  for piecewise constant rates, but none of the previous work acknowledged the coalescent point process representation.
\item Speciation and extinction rates do not depend on time, but the extinction rate may change deterministically as a function of a non-heritable trait with deterministic initial value upon speciation. This is equivalent to an extinction rate changing as a function of the age $x$ of the species ($\lambda(t)=\lambda$, $\mu(t,x)=\mu(x)$).
 The coalescent density is given in Proposition \ref{PropnMh}. This case extends the analyses made in \citep{Lambert2010}. We illustrate this case by providing explicit densities in the case when lifetimes are deterministic (not random) or follow the Gamma distribution with shape parameter 1 or 2. 
\item We demonstrate how the results can be used for processes featuring {\it mass extinction events}, meaning that at some time point $t$ (or multiple time points) in the past each species becomes extinct independently and instantaneously with a fixed probability $p_t$. Again, the reconstructed tree is a CPP, and its coalescent density  is given in Proposition \ref{prop:bottlenecks}. 
In \citep{Stadler2011PNAS}, the probability density of a tree with mass extinction events was derived for piecewise constant speciation and extinction rates, however again the point process representation was not acknowledged.
\end{itemize}

\section*{\textsc{Coalescent point processes}}
%
We now introduce notation and properties of the CPP which will be used later.
Consider a CPP with age $\timestop$ and coalescent density $f$.
From now on, we denote by $H$ a random variable with this coalescent density $f$. We define
$$
F(t) = \frac{1}{P(H> t)}
\qquad t\ge 0,
$$
the \emph{inverse tail distribution} of $H$. One can recover the coalescent density $f$, from $F$
as follows
$$
f(t)=-\frac{d}{dt}P(H>t)=\frac{F'(t)}{F(t)^2}.
$$ 
Let  $N_\timestop$ be the number of extant species in the coalescent point process. Conditional on $N_\timestop=n$, the node depths $H_1,\ldots, H_{n-1}$ are independent copies of $H$ conditioned on $H\le \timestop$.

\subsection*{\textit{Number of lineages}}

The number of lineages in the coalescent point process present at time $s<\timestop$ is exactly one plus 
the number of node depths larger than $\timestop-s$. By independence, except for the ancestral lineage, this number is geometrically distributed with success parameter $P(H> \timestop \mid H > \timestop-s)$. This can be stated as follows.
\begin{prop}
\label{prop:number lineages}
Let $N_s^\star$ denote the number of lineages at time $s$ in the CPP. Then
$$
P(N_s^\star=k) =  P(H> \timestop \mid H> \timestop -s)\,P(H\le \timestop \mid H> \timestop -s)^{k-1} \qquad k\ge 1,
$$
which is the geometric distribution.
In particular,
$$
E(N_s^\star)= \frac{1}{ P(H> \timestop \mid H> \timestop-s)} = \frac{P(H> \timestop-s)}{P(H> \timestop)} = \frac{F(T)}{F(T-s)}.
$$
\end{prop}
In particular, taking $s=T$, since $N_T = N_T^\star$, we obtain the distribution of the number of extant species at time $\timestop$
\begin{equation}
\label{eqn:law N}
P(N_ \timestop =k) =  P(H> \timestop)\,P(H\le \timestop)^{k-1} \qquad k\ge 1.
\end{equation}

\subsection*{\textit{The likelihood of a reconstructed tree}}
 In this section, we display likelihood formulae for trees produced by a coalescent point process. Recall that CPP trees always have at least one tip, so there is no need to condition them upon survival. 
 
 Under a CPP with inverse tail distribution $F$ and coalescent density $f$, the likelihood ${\mathcal L}$ of the non-labelled tree $\tau$ with known topology, stem age $T$, $n$ extant species and node depths $x_1<\cdots< x_{n-1}$  is given by
\begin{equation}
\label{eqn:likelihood1}
{\mathcal L}(\tau,n\mid T)=\frac{C(\tau)}{F(T)}\prod_{i=1}^{n-1} f(x_i) ,
\end{equation}
where $C(\tau) = 1$ if $\tau$ is oriented and $C(\tau) = 2^{n-1-c(\tau)}$ if $\tau$ is non-oriented. 

Note that if $T$ is the \emph{crown age} of $\tau$, that is, if the two longest edges of $\tau$ both have length $T$, then the likelihood ${\mathcal L}_{\text{c}}(\tau)$ (the subscript `c' stands for `crown age') of the reconstructed tree $\tau$ with crown age $T$, $n$ extant species and node depths $x_1<\cdots< x_{n-2}$ (now there are only $n-2$ node depths strictly smaller than $T$), \emph{conditional on speciation at time 0 and survival of the two incident subtrees}, is  the product, properly renormalized, of the likelihoods of the two subtrees conditional on survival, which equals 
\begin{equation}
\label{eqn:likelihood1bis}
{\mathcal L}_{\text{c}}(\tau,n\mid T)=\frac{C(\tau)}{F(T)^2}\prod_{i=1}^{n-2} f(x_i) ,
\end{equation}
where $C(\tau)$ was defined previously. 

Note that if the tree with stem (resp. crown) age $T$ is \emph{conditioned to have exactly $n$ tips}, then the conditioned likelihoods 
become
\begin{equation}
\label{eqn:likelihood2}
{\mathcal L}(\tau\mid T,n)= C(\tau)\prod_{i=1}^{n-1} f_T(x_i) \quad\mbox{ and resp. }\quad {\mathcal L_{\text{c}}}(\tau\mid T,n)= \frac{C(\tau)}{n-1} \prod_{i=1}^{n-2} f_T(x_i), 
\end{equation}
where $f_T(x)\, dx=P(H\in dx\mid H<T)$, that is, $f_T(x) = f(x) F(T)/(F(T)-1)$. Indeed, for the crown age, the probability to have $n$ tips conditional on two ancestors each having alive descendance at $T$ equals $(n-1)P(H<T)^{n-2} P(H>T)^2$. This also reads
\begin{eqnarray}
\label{eqn:likelihood2bis}
{\mathcal L}(\tau\mid T,n)&=& {\mathcal L}(\tau,n\mid T) \left(\frac{F(T)}{F(T)-1}\right)^{n-1}F(T) \notag \\
{\mathcal L_{\text{c}}}(\tau\mid T,n)&=& {\mathcal L_{\text{c}}}(\tau,n\mid T)\,\frac{1}{n-1} \left(\frac{F(T)}{F(T)-1}\right)^{n-2}F(T)^2 
\end{eqnarray}

\subsection*{\textit{Missing tips}}

We will first discuss that reconstructed trees after $p-$ sampling are CPP reconstructed trees if the completely sampled reconstructed trees are CPP reconstructed trees.
Second, $n-$  sampling induces a URT distribution after sampling if the pre-sampling distribution is also URT \citep{Stadler2009SystBiol}. The same holds for diversified sampling as we show below.
Since the number of tips under these two sampling schemes is $n$ with probability $1$, the number of tips does not follow a geometric distribution, and thus under these sampling schemes trees are not CPP. Actually, they are not even CPP conditioned to have $n$ tips, since their $n-1$ node depths are shown to be correlated.
 Last, we show that  higher-level phylogenies obtained from pre-sampling CPP trees give rise to a URT distribution.
We show  that the tree likelihoods under our four schemes of incomplete sampling can be  readily calculated for any model which has a CPP representation under complete sampling.

\paragraph{The $p$-sampling scheme.}
A common way of modeling trees with missing species is to assume that each tip is sampled \emph{independently} with probability $p$ (i.e. Bernoulli sampling). From the orientation of the tree, we know that in a CPP, the coalescence time between species $i$ and species \tanjaR{$j$} is $\max\{H_{i+1},\ldots, H_j\}$. The number of unsampled species between two consecutive sampled species is a geometric r.v. with success probability $p$.

 Therefore, the genealogy of the sample of a CPP is again a CPP, where $H$ is replaced by the r.v. $H_p$ distributed as the maximum of $K$ independent copies of $H$, where $K$ is an independent geometric r.v. with success probability $p$.
 
 As a consequence,
$$
P(H_p\le t)=\sum_{j\ge 1}p(1-p)^{j-1}P(H\le t)^{j} = \frac{pP(H\le t)}{1-(1-p)P(H\le t)}\qquad t>0.
$$
This can be recorded in the following statement. Also, this connection is further developed in \citep{LS}, where the authors study the effect of tip removals, viewed as contemporary extinctions, on the total length of the tree, also called phylogenetic diversity.
\begin{prop}
\label{prop:sampling}
The genealogy of a Bernoulli($p$)-sample taken from a CPP with inverse tail distribution $F$ is a CPP with typical node depth denoted $H_p$ with inverse tail distribution $F_p$ given by
$$
F_p(t) := \frac{1}{P(H_p>t)} =1-p + pF(t). 
$$
\end{prop}
Check that when $p=1$, $1/P(H_p>t)=F(t) = 1/P(H>t)$,  we recover the CPP with typical node depth $H$.

If we are given a complete phylogenetic tree, we can obtain the phylogeny of sampled tips either by 
first reconstructing the phylogenetic tree (i.e., throwing away extinct lineages) and then sampling tips on the reconstructed tree, or by first removing tips from the phylogenetic tree and then reconstructing the sampled tree. This has the following important consequence. For a speciation-extinction model whose reconstructed tree is a CPP with inverse tail distribution $F$, the last proposition implies that the phylogeny of Bernoulli($p$)-sampled species is a CPP with inverse tail distribution $F_p$. Therefore, every property we will state for the complete phylogeny under such models will hold for the incomplete phylogeny (under the Bernoulli($p$)-sampling scheme), provided we change $F$ for $F_p$.

\paragraph{The $n$-sampling scheme.} 
Another way of modeling missing species is to randomly pick $n$ tips out of $N_T\ge n$ tips, by selecting uniformly $n$ tips among $N_T$ (selecting uniformly one tip among $N_T$, then selecting uniformly a second tip among the remaining $N_T-1$, and so on $n$ times). The tree obtained from a CPP after this so-called $n$-sampling scheme is not a CPP any longer. The following results are proved in the Appendix.

\begin{prop}
\label{prop:n-sampling}
The likelihood ${\mathcal L}^s(\tau,m\mid T, n)$ of a reconstructed tree $\tau$ with stem age $T$, $n$ \emph{sampled} species, $m$ missing species (i.e., $n+m$ extant species) and node depths $x_1<\cdots< x_{n-1}$, is given by (writing $x_n=T$)
\begin{equation}
\label{eqn:likelihoodmissing}
{\mathcal L^{\text{s}}}(\tau,m\mid T,n)= {\mathcal L}(\tau,n\mid T) \ {m+n \choose n}\ \sum_{\vec{m}:m_1+\cdots+m_n=m}  \prod_{i=1}^{n}(m_i+1)P(H<x_i)^{m_i}
\end{equation}
where ${\mathcal L}(\tau,n\mid T)$ is given by \eqref{eqn:likelihood1}. 
The same correction factor holds for the likelihood $\mathcal{L^{\text{s}}_{\text{c}}}(\tau\mid T, n)$ of a reconstructed tree $\tau$ with \emph{crown age} $T$, $n$ sampled species, $m$ missing species and node depths $x_1<\cdots< x_{n-2}$, if now we write $x_{n-1}=x_n=T$.

As in \eqref{eqn:likelihood2bis}, the likelihoods ${\mathcal L^{\text{s}}}(\tau\mid T, n, m)$ and ${\mathcal L^{\text{s}}_{\text{c}}}(\tau\mid T, n, m)$ \emph{conditional} on the total number $n+m$ of extant species are given by
\begin{equation}
\label{eqn:likelihood3}
{\mathcal L^{\text{s}}}(\tau\mid T,n,m)= {\mathcal L^{\text{s}}}(\tau,m\mid T,n) \left(\frac{F(T)}{F(T)-1}\right)^{n+m-1}F(T)
\end{equation}
and
\begin{equation}
\label{eqn:likelihood3bis}
{\mathcal L^{\text{s}}_{\text{c}}}(\tau\mid T,n, m)= {\mathcal L^{\text{s}}_{\text{c}}}(\tau,m\mid T,n)\,\frac{1}{n+m-1} \left(\frac{F(T)}{F(T)-1}\right)^{n+m-2}F(T)^2. 
\end{equation}

\end{prop}
A consequence of the last proposition is that node depths after $n-$sampling are not iid any more, conditionally or not on the total number $m$ of tips. Thus, after  $n-$sampling, CPP trees are not CPP any longer.

When $m$ is large, the right-hand side in \eqref{eqn:likelihoodmissing} is hardly computable. There is no simpler formula available for this likelihood. Nevertheless, we are able to provide a quite simple formula for the multivariate distribution function of the node depths of the tree, as we now show.

Relabel the $n$ sampled tips $1, 2,\ldots, n$ ranked in the same order as they were in the initial coalescent point process and set $H_i'$ the coalescence time between sampled tip $i$ and sampled tip $i+1$, $i=1,\ldots, n-1$. 
Setting
$$
p_0:= P(H<T)\quad\mbox{ and }\quad p_i:= P(H<x_i)\quad i=1,\ldots, n,
$$
if $x_1, \ldots, x_{n-1}$ are pairwise distinct, we get
\begin{multline}
\label{eqn:n-sampling}
P(N_T=n+m, H_1' <x_1, \ldots, H_{n-1}'<x_{n-1}) \\= {m+n \choose n}\ (1-p_0)p_1\cdots p_{n-1}\ 
\left(  
\sum_{i=1}^{n-1} \frac{p_i^{m+n}}{(p_i-p_0)^2\prod_{j=1,\ldots, n-1, j\not=i} (p_i-p_j)}\right.\\
\left. -
\sum_{i=1}^{n-1} \frac{p_0^{m+n}}{(p_0-p_i)\prod_{j=1}^{n-1} (p_0-p_j)}
+
 \frac{(n+m) p_0^{n+m-1}}{\prod_{j=1}^{n-1} (p_0-p_j)} 
\right).
\end{multline}
This shows in particular that the node depths $H_1',\ldots, H_{n-1}'$ have the same distribution, given by
\begin{equation}
\label{eqn:n-sampling1}
P(N_T=n+m, H_1'<x_1) = n(1-p_0)p_0^{n-2}p_1(p_0-p_1)^{-n} \int_{p_1}^{p_0} y^m (y-p_1)^{n-1}\, dy.
\end{equation}

\paragraph{Diversified sampling.}
Diversified sampling, defined as picking $n$ tips such that the most ancient speciation events are kept \citep{Stadler2011MBE-Sampling}, is not a CPP, as we now show. 
 Diversified sampling essentially means that we pick the $n-1$ deepest nodes ($n-1$ first order statistics of $(H_i)$) yielding the following results (see last paragraph and also \citep{Stadler2011MBE-Sampling}, p.2581, bottom left equation). The likelihood ${\mathcal L}^\text{d}(\tau,m\mid T, n)$ of a reconstructed tree $\tau$ with stem age $T$, $n$ sampled species,  $m$ missing species (i.e., $n+m$ extant species), and node depths $x_1<\cdots< x_{n-1}$, is given by 
$${\mathcal L^{\text{d}}}(\tau,m\mid T,n)= {\mathcal L}(\tau,n\mid T)\ {m+n-1 \choose n-1} \ P(H<x_1)^m.
$$
The likelihood  ${\mathcal L}^\text{d}(\tau\mid T, n,m)$ \emph{conditional} on the total number of missing tips is obtained as usual from ${\mathcal L^{\text{d}}}(\tau,m\mid T,n)$ by dividing it by $P(H<T)^{n+m-1} P(H>T)$.
Clearly, the likelihood cannot be factorized as a product of identical terms, and so after diversified sampling, CPP trees are not CPP any longer, conditionally or not on the total number $m$ of tips.
However, each permutation of branching times is equally likely, meaning the tree distribution obtained from pre-sampling CPP trees  is URT. We highlight that diversified sampling operates directly on trees ignoring branch lengths. Thus, as CPP trees induce URT,  we  showed that pre-sampling URT induces URT after diversified sampling.  

%

\paragraph{Higher-level phylogenies.} \label{SecHigherOrder}
In higher-level phylogenies, not all species are included in a reconstructed phylogeny, because some monophyletic clades are collapsed into one tip, with this tip having the number of tips in the original subtree assigned (numbers of species in a clade). Here we assume each lineage present at time $x_{cut}$ in the past is collapsed into one tip representing a clade of size $k_i$ ($i=1,\ldots,n$ for a higher-level phylogeny on $n$ tips), and we define $k=\sum_{i=1}^{n} k_i$. 
Previously only likelihood inference methods assuming constant speciation and extinction rates were available \citep{Paradis2003,Stadler2012groups}. However, the CPP representation facilitates the calculation of the tree likelihood, which is given by
\begin{equation}
\label{eqn:likelihoodHO}
\mathcal{L}^\text{hl}(\tau,k,n\mid T) = {\mathcal L}(\tau,n\mid T)\ P(H<x_{cut})^{k-n}.
\end{equation}

The likelihood $\mathcal{L}^\text{hl}(\tau\mid T,n,k)$ conditional on $n$ clades and $k$ extant species is obtained from $\mathcal{L}^\text{hl}(\tau,k,n\mid T)$ by dividing it by $P(H<x_{cut})^{k-n} P(x_{cut}<H<T)^{n-1} P(H>T) {k-1 \choose n-1}$.

The likelihood $\mathcal{L}^\text{hl}(\tau, k\mid T,n)$ conditional on $n$ extant lineages at depth $x_{cut}$ (clades) is obtained from $\mathcal{L}^\text{hl}(\tau,k,n\mid T)$ by dividing it by 
\begin{eqnarray*} & &
\sum_{k=n}^\infty  P(H<x_{cut})^{k-n} P(x_{cut}<H<T)^{n-1} P(H>T) {k-1 \choose n-1} \\
&=& P(H<T\mid H>x_{cut})^{n-1} P(H>T \mid H>x_{cut}),
\end{eqnarray*} 
which indeed is the probability to have $n$ extant lineages at depth $x_{cut}$, according to Proposition \ref{prop:number lineages}. Note that the likelihood in \eqref{eqn:likelihoodHO} can be written in product form, so that the pairs constituted by node depth and clade size are iid random numbers. 
Furthermore, each permutation of branching times is equally likely, meaning the higher-level phylogeny distribution obtained from pre-sampling CPP trees (ignoring the tip labels) is URT.

\section*{\textsc{Main results}}

In this section, we will prove the statements of Table 1.

\subsection*{\textit{Model classes 4-5 do not induce URT}}

Here, we first give a counter-example of a model in class 4 (trait-dependent speciation rate) which does not induce URT, even in the absence of extinction (zero extinction rate). The trait under consideration is the age, which is a non-heritable trait. Since non-heritable traits are particular cases of heritable traits, this counter-example is sufficient to prove that model class 4 does not induce URT. Then we give a counter-example of a model in class 5 (heritable trait-dependent extinction rate) which does not induce URT, even when the speciation rate is constant.

Let us start with model class 4.
Suppose speciation happens deterministically in each species once it reaches age 1. The resulting reconstructed tree is a so-called caterpillar tree, i.e. a tree where each speciation event has only a single species descending to the left and all other species descending to the right. This means that the caterpillar tree has probability 1 and all other ranked trees have probability 0, which is obviously different from a uniform distribution on ranked trees. This counter-example does not rigorously fit our general model, since the speciation rate is infinite at age 1, but can be modified as follows. If the speciation rate is positive (and finite) inside an arbitrarily small time window around age 1 and zero outside, most trees will only have one extant species at time $T$, but conditional on having $n$ species extant at $T$, the probability of a caterpillar tree can be arbitrarily close to 1 (see Figure \ref{fig:counterexamples}a). 

In the case of model class 5, we can also produce reconstructed trees which are caterpillar trees with a high probability. Now species can be of two types, long-lived (type 0, extinction rate 0), or short-lived (type 1, extinction rate $\mu$). The trait under consideration is the pair $(i,a)$, where $i$ is the type of the species and $a$ its age. The speciation rate is constant. The inheritance is as follows: if the age of the mother species upon speciation is close to 1 (in the sense of the previous example), and if this species is of type 0, then the type of the incipient species is 0. In all other cases, the type of the incipient species is 1. Age of the incipient species is set to 0 as usual. If $\mu$ is large, then in such a model, most trees only have one extant species at time $T$. In (the rare) trees with more than 1 tip, all tips will be of type 0, born from species of type 0 at age 1, with a high probability. This shows that reconstructed trees with a fixed number of tips will be caterpillar trees with an arbitrarily  high probability, as $\mu$ gets large enough (see Figure \ref{fig:counterexamples}b).

\subsection*{\textit{Model classes 1-3 induce URT}}

We want to show that when the speciation rate and the extinction rate possibly both depend on (i) time and (ii) number of species, and that the extinction rate possibly further depends on (iii) a non-heritable trait, then the distribution induced on ranked tree shapes by the reconstructed tree is URT. This property holds even if the distribution of the non-heritable trait at birth depends on the number of coexisting species. \cite{Aldous2001} proved the statement for (i) and (ii), and \cite{Stadler2013JEB} provided a non-formal argument for (iii).

The formal argument relies on the fact that at each branching event, the lineage on which sprouts the incipient lineage is chosen uniformly among existing lineages. Then since further extinction and speciation events do not depend on the orientation of the tree, regrafting subtrees of the complete tree on other lineages but at the same time, does not change the probability of the complete tree. This property obviously carries over to the reconstructed tree, which has the following consequence. 

The reconstructed tree is an oriented, ultrametric tree with $n$ tips and node depths $H_1$, \ldots, $H_{n-1}$. Let $\tau_i$ be the subtree descending from the $i$-th branch, that is the tree spanned by tips $i, i+1,\ldots, \sigma(i, H_i)-1$, where $\sigma(i,x):=\min\{k\in\{i+1,\ldots,n\}: H_k >x\}$ (with the convention $H_n=+\infty$). The tree obtained after regrafting $\tau_i$ on the $j$-th branch (provided $H_j>H_i$) is the oriented tree whose node heights remain in the same order, except that the block $(H_i, H_{i+1},\ldots, H_{\sigma(i, H_i)-1})$ has been inserted between branches $\sigma(j, H_i)-1$ and $\sigma(j, H_i)$. According to the previous paragraph, this oriented tree has the same distribution as the initial reconstructed tree. 

By composing several such subtree regraftings, we easily see that we can perform any permutation on edges without changing the probability of the oriented tree. Therefore, each oriented tree has the same probability, which induces URT on the (unoriented) reconstructed tree.



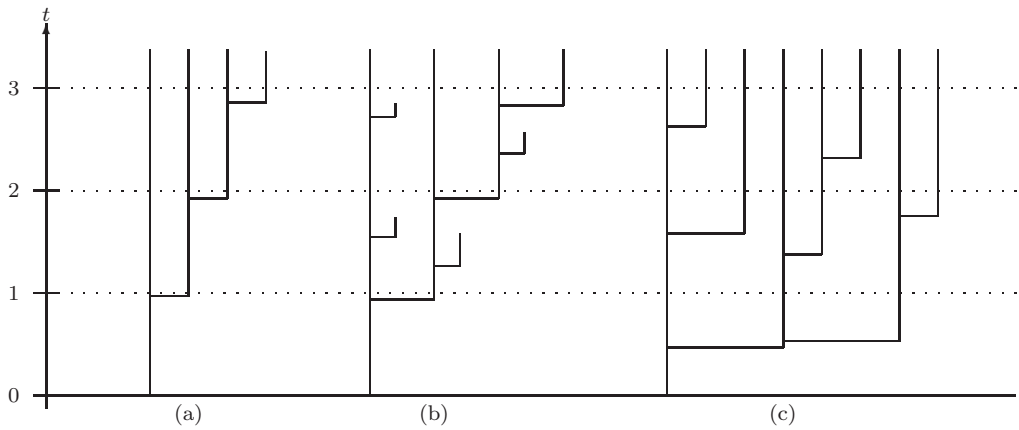
\begin{figure}[!ht]
\unitlength 1.7mm 
\linethickness{0.4pt}
\ifx\plotpoint\undefined\newsavebox{\plotpoint}\fi 
\begin{picture}(79,33.75)(-6,0)
\put(3,4){\line(1,0){2}}
\put(3,12){\line(1,0){2}}
\put(3,20){\line(1,0){2}}
\put(3,28){\line(1,0){2}}
\put(29,4){\line(0,1){27}}
\put(29,11.5){\line(1,0){5}}
\put(34,11.5){\line(0,1){19.5}}
\put(34,19.375){\line(1,0){5}}
\put(39,19.375){\line(0,1){11.625}}
\put(39,26.625){\line(1,0){5}}
\put(44,26.625){\line(0,1){4.375}}
\put(34,14.125){\line(1,0){2}}
\put(36,14.125){\line(0,1){2.5}}
\put(29,16.375){\line(1,0){2}}
\put(31,16.375){\line(0,1){1.5}}
\put(29,25.75){\line(1,0){2}}
\put(31,25.75){\line(0,1){1}}
\put(39,22.875){\line(1,0){2}}
\put(41,22.875){\line(0,1){1.625}}
\put(52,4){\line(0,1){27}}
\put(52,25){\line(1,0){3}}
\put(55,25){\line(0,1){6}}
\put(52,7.75){\line(1,0){9}}
\put(61,7.75){\line(0,1){23.25}}
\put(61,8.25){\line(1,0){9}}
\put(70,8.25){\line(0,1){22.75}}
\put(61,15){\line(1,0){3}}
\put(64,15){\line(0,1){16}}
\put(64,22.5){\line(1,0){3}}
\put(67,22.5){\line(0,1){8.5}}
\put(70,18){\line(1,0){3}}
\put(73,18){\line(0,1){13}}
\put(4,3){\vector(0,1){30}}
\put(12,4){\line(0,1){27}}
\put(12,11.75){\line(1,0){3}}
\put(15,11.75){\line(0,1){19.25}}
\put(15,19.375){\line(1,0){3}}
\put(18,19.375){\line(0,1){11.625}}
\put(18,26.875){\line(1,0){3}}
\put(21,26.875){\line(0,1){4}}
\put(4,4){\line(1,0){75}}
\multiput(3.93,11.93)(.986842,0){77}{{\rule{.4pt}{.4pt}}}
\multiput(3.93,19.93)(.986842,0){77}{{\rule{.4pt}{.4pt}}}
\multiput(3.93,27.93)(.986842,0){77}{{\rule{.4pt}{.4pt}}}
\put(52,16.625){\line(1,0){6}}
\put(58,16.625){\line(0,1){14.375}}
\put(1.5,4){\makebox(0,0)[cc]{\scriptsize $0$}}
\put(1.5,12){\makebox(0,0)[cc]{\scriptsize $1$}}
\put(1.5,20){\makebox(0,0)[cc]{\scriptsize $2$}}
\put(1.5,28){\makebox(0,0)[cc]{\scriptsize $3$}}
\put(4,33.75){\makebox(0,0)[cc]{\scriptsize $t$}}
\put(15,2.5){\makebox(0,0)[cc]{\scriptsize (a)}}
\put(34,2.5){\makebox(0,0)[cc]{\scriptsize (b)}}
\put(61,2.5){\makebox(0,0)[cc]{\scriptsize (c)}}
\end{picture}
\caption{Trees illustrating the counterexamples for (a) model class 4; (b) model class 5; (c) model classes 2-3. As in main text, (a) reconstructed trees in model class 4 are not URT: age-dependent speciation rates can produce caterpillar trees with high probability (w.h.p.); (b) reconstructed trees in model class 5 are not URT: heritable trait-dependent extinction rates can produce caterpillar trees w.h.p. (short edges are of type 1 and long edges are of type 0); (c) reconstructed trees in model class 2-3 are not CPP: rates dependent on the number of species can produce trees where the first and second speciations are arbitrarily close w.h.p.}
\label{fig:counterexamples}
\end{figure}

\subsection*{\textit{Model classes 2-3 do not induce CPP}}

Here, we show that when extinction rates (model class 2) or speciation rates (model class 3) depend on the number of species, reconstructed trees cannot be CPP in general.

We first provide a counter-example in model class 2. Assume that the speciation rate is constant equal to $\lambda$, and that the extinction rate is equal to $\mu$ when the number of species is two, and to 0 otherwise. If $\mu$ is large, then in such a model, most trees only have one extant species at time $T$. In (the rare) trees with two tips, the unique branching time of the reconstructed tree is close to $T$ with high probability. In (the rare) trees with more than two tips, the first branching time of the reconstructed tree is the first speciation event followed sufficiently closely by a second speciation event to have three co-existing species before the rapid extinction caused by the co-existence of two species. Therefore, in such trees, the second branching time closely follows the first branching time with high probability, whereas further branching times are spaced as in a Yule process. This shows that the reconstructed tree has correlated node depths, and therefore is not a CPP (see Figure \ref{fig:counterexamples}c).

We now treat model class 3. Recall that the number of speciation events in a CPP follows a geometric distribution. In particular, the probability of having $K$ speciation events is strictly positive for any $K$. However, 
this will not be the case when speciation rate is zero as soon as the number of species is greater than $K$. 
Thus, reconstructed trees under models of class 3 cannot be CPP in general. However, one could wonder whether they keep the property of CPP that conditional on the number of tips, node depths are independent with the same distribution (as is the case in the previous example where the speciation rate is constant for less than $K$ co-existing species). But there is a counter-example very similar to the one displayed for model class 2. Assume that the extinction rate is zero and that the speciation rate is equal to $C\lambda$ when the number of species is 2, and to $\lambda$ otherwise. If $C$ is large, then  as earlier, the first and second branching times of the reconstructed tree will be closer than others, showing that node depths are correlated (see Figure \ref{fig:counterexamples}c).

\subsection*{\textit{Model class 1 induces CPP}}

Here, we consider a diversification model in class 1, that is, the extinction rate may depend on a non-heritable trait and on absolute time, and the speciation rate may depend on absolute time (only).

More specifically, the instantaneous speciation rate at time $t$ is denoted $\lambda(t)$ for all species and the instantaneous extinction rate at time $t$ of a species carrying trait value $x$ at this time is denoted $\mu(t,x)$. The fact that the trait is not heritable means that upon speciation at time $t$, the trait of the daughter species is drawn from a distribution $\nu_t(dx)$ that may depend on $t$, and that the trait of the mother species is not altered by the speciation event. In addition, traits of different species, conditional on the trait values at speciation, change independently through time according to the same (possibly stochastic, possibly time-inhomogeneous) dynamics. Age is a typical example of a non-heritable trait. For simplicity we will assume that a trait is always one-dimensional.
 
In order to prove future statements in this section, we need to remind the reader of some mathematical properties of instantaneous rates.

\subsubsection*{Precisions about rates}

The meaning of rate has to be taken in the usual mathematical sense. To say that the speciation rate at time $t$ is $\lambda(t)$ is equivalent to saying that a given species extant at time $t$ gives birth to a daughter species in the time interval $(t, t+h)$ with probability $h\lambda(t)+o(h)$ as $h$ goes to 0. Equivalently, the number of new species born from the same species during the time interval $(a,b)$ follows a Poisson distribution with parameter $\int_a^b\lambda(t)\, dt$. In particular, a species extant during the time interval $(a,b)$ does not speciate during this interval with probability
$$
\exp\left(-\int_a^b \lambda(t)\, dt\right).
$$
Similarly, a species carrying trait value $x$ at time $t$ becomes extinct during the time interval $(t, t+h)$ with probability $h\mu(t,x)+o(h)$ as $h$ goes to 0. Now if a species born at time $a$ carries trait value $X_t$ at time $t$, for $t\in (a,b)$, then the probability to not become extinct before time $b$ equals
$$
\exp\left(-\int_a^b \mu(t, X_t)\, dt\right).
$$
To obtain the probability that a species born at time $a$ with trait $x$ survives at least until time $b$, the last quantity has to be averaged over all trait dynamics with initial starting point $X_a=x$ (see forthcoming Equations \eqref{eqn:hts} and \eqref{eqn:ustx}).

\subsubsection*{Characterizing the CPP under model class 1}
We start with one species at time 0, we condition the tree to have at least one extant species at time $T$, and we label $0, 1, \ldots, N_T-1$ the $N_T$ species extant at time $T$ from left to right, assuming that the tree is oriented. Recall that the node depths of the reconstructed tree are denoted $H_0=T$ and then $H_1, \ldots, H_{N_T-1}$, where $H_i$ is the coalescence time between species $i-1$ and species $i$ (see Figure \ref{fig:examplePhylo}).
The proof of our main result below is put to the Appendix.
\begin{thm} \label{ThmMain}
Consider a macro-evolutionary tree generated by a model of class 1 (notation specified above), started at 0 and conditioned on having at least one species extant at time $\timestop$. 
The oriented reconstructed tree is a coalescent point process with typical node depth $H$ whose inverse tail 
distribution is given by
$$
F(t):=\frac{1}{P(H>t)} =\exp\left( \int_{\timestop-t}^ \timestop \lambda(s)\, (1-q(s))\,ds\right)\qquad t\in[0,T],
$$
where $q(t)$ denotes the probability that a species born at time $t$ has no descendants by time $\timestop$.
\end{thm}
 The quantity $q(t)$ involved in the expression of $F$ is not directly available from the model parameters (but see Appendix, Section ``Generator and Feynman-Kac formulae''). In the next statement, we give a characterization of $F$ in terms of quantities which are more easily computable from the model parameters. The proof only requires a few lines and is found in the Appendix.
\begin{prop}
\label{prop:NEW}
For any $s\ge t$, let $g(t,s)$ be the density at time $s$ of the extinction time of a species born at time $t$.  Then $F$ is the unique solution to the following integro-differential equation
\begin{equation}
\label{eqn:hidden convolution}
F'(t) = \lambda(T-t)\,\left( F(t) - \int_0^t ds\ F(s)\,g(T-t,T-s)\right)\qquad t\ge 0,
\end{equation}
with initial condition $F(0)=1$.
\end{prop}
Recall that for a CPP with inverse tail distribution $F$, the reconstructed tree of sampled  tips, when tips  are sampled independently with probability $p$, is a coalescent point process with node depths distributed as $H_p$ and inverse tail distribution $F_p$ given by $F_p= 1-p+pF$, where $F$ can be computed thanks to one of the previous two statements.  Also recall (or check) that the common density of node depths, or coalescent density, is $F'/F^2$ (or $F_p'/F_p^2$ in case of sampling) so that the knowledge of $F$ yields instantaneously the likelihood of a reconstructed tree produced by a macro-evolutionary model belonging to class 1, thanks to the results in Section``The likelihood of a reconstructed tree''.

This proposition has the very important following consequence. From the knowledge of $g$, the pair $(F,F')$ can be computed by (possibly numerical) integration of \eqref{eqn:hidden convolution}, and the coalescent density is then obtained as $F'/F^2$. This represents an important advance, because $g$ can be made available in terms of the model parameters much more easily than $q$, as we now see.

Let $X$ denote the stochastic process which describes the dynamics of the trait in a single species. 
Invoking arguments from Section ``Precisions about rates'' the density at time $s$ of the extinction time of a species born at time $t$ with trait value $x$, conditional on the trait dynamics ($X_t =x$ in particular), equals
$$
\mu(s, X_{s})\ e^{-\int_t^s dr\, \mu(r,X_{r})}.
$$
If $\EE_{t,x}$ denote the expectation associated to the distribution of $X$ started at time $t$ in state $x$, then  
\begin{equation}
\label{eqn:hts}
g(t,s) = \int_\RR \nu_t (dx)\ u_s(t,x) \qquad s\ge t,
\end{equation}
where
\begin{equation}
\label{eqn:ustx}
u_s(t,x):=\EE_{t,x} \left(\mu(s, X_{s})\ e^{-\int_t^s dr\, \mu(r,X_{r})}\right)\qquad s\ge t.
\end{equation}
Now assume that $X$ is a Markov process. Then in general, $u_s$ (and so $g$) can be computed thanks to the Feynman-Kac formula, which ensures that $u_s$ is the unique solution to 
\begin{equation}
\label{eqn:Feynman-Kac}
\frac{\partial u_s}{\partial t} (t,x) + L_t u_s (t,x) =   \mu(t,x)\ u_s(t,x),
\end{equation}
with terminal condition $u_s(s,x)= \mu(s,x)$, and where $L_t$ is the generator at time $t$ of $X$ (see Appendix, Section ``Generator and Feynman-Kac formulae''). Specifically, when $X$ is the age, the initial trait value is $x=0$ and the age at $s$ of a species born at $t$ is $X_s= s-t$ so that 
\begin{equation}
\label{eqn:FK age}
g(t,s) = \mu(s, s-t)\ e^{-\int_t^s dr\, \mu(r,r-t)} \qquad s\ge t.
\end{equation}
 In the following section, we display some special cases of biological interest leading to at least partially explicit expressions for $g$ and for $F$.

\section*{\textsc{Three special cases}}

In this section, we study three special cases of macroevolutionary models in class 1. The reconstructed trees under each of these models are  coalescent point processes as shown in the previous section. We study these coalescent point processes, and provide means of computing their associated coalescent density. The three cases are: trait-independent models (Markovian case), time-independent models and mass extinction events. 

\subsection*{\textit{The time-dependent models, without trait dependency}}

Here, we assume that speciation and extinction rates may depend on time but do not depend on a trait. Thus, we denote by $\lambda(t)$ the speciation rate at time $t$ and by $\mu(t)$ the extinction rate at time $t$. We also define 
$$
r(t)=\lambda(t) - \mu(t)
$$
sometimes referred to as the \emph{time-dependent diversification rate}.
Using Proposition \ref{prop:NEW},  we get the following statement by a few lines of calculations  which are put to the Appendix.




\begin{prop}
\label{prop:markov}
In the case when the  rates $\lambda$ and $\mu$ only depend on time, the reconstructed tree is a CPP whose inverse tail distribution $F$ is given by
$$
F(t) =1+\int_{T-t}^T ds\,\lambda(s)\,e^{\int_s^T du \,r(u)} .
$$
\end{prop}
Recall that the inverse tail distribution of the incomplete phylogeny with sampling probability $p$ is $F_p = 1-p+pF$, so that
$$
F_p(t) =1+p\int_{T-t}^T ds\,\lambda(s)\,e^{\int_s^T d u\,r(u)} .
$$
When rates do not depend on time, the diversification process is a linear birth--death process with birth rate $\lambda$ and death rate $\mu$. The last formula then boils down to
\begin{equation}
\label{eqn:Markovian scale}
F_p(t)=\begin{cases}
1 + \frac{p\lambda}{r}\big(e^{rt}- 1\big) & \text{if } r\not=0 \\
1+p\lambda t & \text{if }r=0.
\end{cases}
\end{equation}
Let us now check that we can recover the likelihood formulae for the reconstructed tree of the birth--death model computed by different means in the earlier works \citep{Stadler2009JTB, Stadler2010JTB,hallinan2012generalized}.
Recall that the coalescent density $f_p$ is given by $f_p=F_p'/F_p^2$, so in the case $r\neq 0$, we further obtain
$$
f_p(t)=
 \frac{p \lambda r^2  e^{-rt}}{(p \lambda +(r-p\lambda)e^{-rt })^2   }
$$
and
$$
P(H_p<\timestop) = 1-\frac{1}{F_p(\timestop)} = \frac{p \lambda (1- e^{-r\timestop })}{p \lambda +(r-p\lambda)e^{-r\timestop } }.
$$
Plugging these expressions into the likelihood formulae \eqref{eqn:likelihood2} and \eqref{eqn:likelihood1} respectively yields
%
%
%
\citep{Stadler2009JTB}, Equation (2), and
%
\citep{Stadler2010JTB} Corollary 3.7 (with $\psi=0,m=0,k=0$).

\subsection*{\textit{Deterministic non-heritable trait dynamics without time-dependence}}

\subsubsection*{Main result on age-dependent extinction rates}

Here, we assume that (a) rates are not time-dependent, (b) the dynamics of the trait is deterministic, and (c) the initial value of the trait of a new species is also deterministic. Because of assumption (a), we can denote by $\lambda$ the speciation rate and by $\mu(x)$ the speciation rate of a species carrying trait value $x$. Because of assumptions (b) and (c), the trait of a species is a deterministic function, say $\phi$, of its age $a$, so that the death rate of a species of age $a$ is $\mu(\phi(a))$. It is then equivalent to assume that the death rate is a function $\tilde{\mu} = \mu\circ\phi$ of the species age. Sticking to the notation $\mu$ instead of $\tilde{\mu}$, we get that $g(t,s) = g(t-s)$, where $g$ is now the density of the lifetime of a species. More specifically, let $L$ denote the lifetime of a species, i.e., the age at which a species becomes extinct. 
Invoking arguments from Section ``Precisions about rates'', the probability of becoming extinct before age $a$ is
$$
P(L<a) = \int_0^a g(s)\, ds = 1-e^{-\int_0^a ds\, \mu(s)} ,
$$
so by differentiating,
\begin{equation}
\label{eqn:density from rate}
g(a)=\mu(a)\ e^{-\int_0^a ds\, \mu(s)} \qquad a\ge 0.
\end{equation}
Note that the integral $\int_0^\infty g(a)\, da$ is equal to the probability  $1-e^{-\int_0^\infty ds\, \mu(s)}$ that $L$ is finite, and so can be strictly smaller than 1. Conversely, if a species lifetime has density $g$, then the extinction rate is the following function of age
\begin{equation}
\label{eqn:rate from density}
\mu(a) = \frac{g(a)}{1-\int_0^a g(s)\, ds} \qquad a\ge 0,
\end{equation}
which is constant only if the density  $g$ of the species lifetime $L$ is exponential.
\begin{prop} 
\label{PropnMh}
In the case when the extinction rate $\mu$ is a function of age, but both $\lambda$ and $\mu$ are time-independent, the reconstructed tree is a CPP whose inverse tail distribution $F$ is the unique solution to 
\begin{equation}
\label{eqn:convolution}
F'(t) = \lambda\,\left( F(t) - F\star g(t)\right),
\end{equation}
with $F(0)=1$, where $g$ is the density of species lifetimes, specified by \eqref{eqn:density from rate}, and $\star$ denotes the convolution product. Equivalently, $F$ is the unique non-negative function with Laplace transform 
\begin{equation}
\label{eqn : LT scale}
\int_0^\infty F(t) \,e^{-tx}\, dt = \frac{1}{\psi(x)},
\end{equation}
where 
$$
\psi(x) = x-\lambda+ \lambda \int_0^\infty g(t)\,e^{-tx}\,dt\qquad x\ge 0.
$$
\end{prop}
The proof of this proposition can be found in the Appendix. Note that the inversion of Laplace transforms can sometimes be numerically unstable, and it can then be preferrable to use the convolution equation \eqref{eqn:convolution} to obtain numerical evaluations of $F$.

Proposition \ref{PropnMh} is proved by other means in \citep{Lambert2010}, where the contour process of the phylogenetic tree is defined. This process starts at the extinction time of the progenitor species, decreases linearly, and makes jumps at each encounter of a speciation event, whose size is the lifetime of the incipient species. 
The contour process of the tree truncated at $T$ is a L{\'e}vy process reflected below $T$ and killed upon hitting 0. In the jargon of stochastic processes, the function $\psi$ is called the Laplace exponent of this L{\'e}vy process and $F$ is called its scale function. More information (on dead branches) than that on the reconstructed tree can be retrieved from the knowledge of this contour process, but we will not develop this point here. 

We will now discuss some special cases of $\mu(x)$.

\subsubsection*{Deterministic lifetimes}
\label{subsec:deterministic}
The result of the previous section holds even when species lifetimes do not have a density. As an example, we now treat the case of a fixed species lifetime equal to $b$ with probability 1. This amounts to replacing the distribution $g(a)\, da$ by $\delta_b(da)$, i.e., the Dirac measure at $b$. Thanks to \eqref{eqn : LT scale}, $F$ is then the unique non-negative function whose Laplace transform equals
$$
\int_0^\infty F(t) \,e^{-tx}\, dt = \left( x-\lambda+ \lambda e^{-bx}
\right)^{-1},
$$ 
or equivalently thanks to \eqref{eqn:convolution}, it is the unique solution to $F(0) =1$ and
$$
F'(t) = \lambda\,\left( F(t) - F(t-b)\right).
$$
In particular $F$ has a continuous derivative (except at $b$) and can be computed as follows. For any integer $n$, for any $t\in[nb, (n+1)b]$
$$
F(t) = P_n(\lambda e^{-\lambda b}(t-nb))\, e^{\lambda t},  
$$
where $P_n$ is a polynomial of degree $n$ solving the recurrence relationship 
$$
P_{n+1}(t) = P_n (B) - \int_0^t P_n(s)\, ds\qquad t\in [0,B],
$$
with $B=b\lambda e^{-b\lambda}$, and initial condition $P_0 \equiv 1$. The polynomials $P_n$ can be evaluated in a straightforward manner by any software of symbolic calculus (e.g., Mathematica). For our purpose, it is even sufficient to compute $P_n$ for the integers $n$ such that $nb \le T$ (since we only require $F(t)$ for $t \leq T$).

\subsubsection*{Exponentially distributed lifetimes}

In the case when the species lifetimes are exponentially distributed with parameter $\mu$ (i.e. $\mu$ does not depend on age of the species), the diversification process is a linear birth--death process with birth rate $\lambda$ and death rate $\mu$, and we should recover the expression given by Equation \eqref{eqn:Markovian scale}. Indeed, it is easy to obtain $\psi(x)=x(x-r)/(x+\mu)$, where $\psi$ is defined in Proposition \ref{PropnMh} and $r=\lambda-\mu$ is the net diversification rate. It is then straightforward to invert the Laplace transform in \eqref{eqn : LT scale}, which yields Equation \eqref{eqn:Markovian scale}, as expected.

\subsubsection*{Gamma distributed lifetimes}
Here, we assume that species lifetimes follow a Gamma distribution with shape parameter 2, i.e., their probability density is $g(a) = \theta^2 a\, e^{-\theta a}$. The parameter $\theta$ is not an extinction rate any longer, since from \eqref{eqn:rate from density}, the age-dependent extinction rate is given by
$$
\mu(a) = \frac{g(a)}{\int_a^\infty g(s)\, ds} = \frac{\theta^2 a}{1+\theta a} ,
$$
which increases from 0 for small ages to $\theta$ for large ages. 



It is straightforward to compute the function $\psi$ defined in Proposition \ref{PropnMh}
$$
\psi(x) = \frac{x\, Q(x)}{\left(x+\theta\right)^2} \qquad x\ge 0,
$$
where
$$
Q(x) = x^2 +(2\theta -\lambda) x +\theta(\theta-2\lambda).
$$
Provided that $\theta\not=2\lambda$ (and that both parameters are nonzero), $Q(x) = (x-x_1)(x-x_2)$, where $x_1<x_2$ are both nonzero, and given by
$$
x_1 = \frac{\lambda - 2\theta  -\sqrt{\Delta}}{2}\quad \mbox{ and }\quad x_2 = \frac{\lambda - 2\theta  +\sqrt{\Delta}}{2},
$$
with $\Delta = \lambda^2 + 4\lambda \theta$. Then $1/\psi$ can be decomposed as follows
$$
\frac{1}{\psi(x)} = \frac{\alpha}{x-x_1} + \frac{\beta}{x-x_2} + \frac{\gamma}{x},
$$
where
$$
\alpha = -\frac{(\lambda-\sqrt{\Delta})^2}{4x_1 \sqrt{\Delta}}, \qquad
\beta = \frac{(\lambda+\sqrt{\Delta})^2}{4x_2 \sqrt{\Delta}},\qquad
\gamma = \frac{\theta}{\theta-2\lambda}.
$$
Note that there are the following alternative formulae for $\alpha$ and $\beta$
$$
\alpha = \frac{\lambda(\lambda+\theta-\sqrt{\Delta})}{(\theta-2\lambda)\sqrt{\Delta}} \quad\mbox{ and }\quad
\beta = -\frac{\lambda(\lambda+\theta+\sqrt{\Delta})}{(\theta-2\lambda)\sqrt{\Delta}} .
$$
It is then elementary to invert \eqref{eqn : LT scale} to obtain $F$, and thanks to $F_p = 1-p + p F$,
$$
F_p(t) = 1-p + \gamma p + p\alpha \,e^{tx_1} + p\beta \,e^{tx_2}.
$$
Since $f_p = F_p'/F^2$, we get
$$
f_p(t) = \frac{p x_1 \alpha \,e^{tx_1} + p\beta x_2\,e^{tx_2}}{\left(1-p + \gamma p + p\alpha \,e^{tx_1} + p\beta \,e^{tx_2}\right)^2} .
$$

\subsection*{\textit{Mass extinction events}}

We again start with a phylogenetic tree running between times $0$ and $\timestop$ and add extra extinctions at fixed times $\timestop-s_k<\ldots<\timestop-s_1$ by assuming that each  lineage is independently terminated (together with its subsequent descendance) at time $\timestop-s_i$ with the same fixed probability $1-\varepsilon_i$, as for so-called bottlenecks in population genetics. For example, a single lineage starting at time $0$ and ending up at time $\timestop$ survives the $k$ mass extinction events and makes it to time $\timestop $ with probability $\prod_{i=1}^k \varepsilon_i$. Bernoulli sampling with probability $p$ can be seen as a special case of bottleneck at time $\timestop $, with $s_0=0$ and  (survival probability) $\varepsilon_0=p$. 

Notice that the effect of mass extinctions on the phylogeny of contemporaneous species is the same on the initial phylogeny as on the smaller tree which is the reconstructed tree obtained in the absence of mass extinctions. Then 
instead of working with the forward-in-time diversification process, we can as well work with the associated coalescent point process obtained before the passage of bottlenecks. Therefore, the following proposition can be applied to any class 1 model of diversification, provided the function $F$ is chosen to be the inverse tail distribution of the associated CPP reconstructed tree. The following proposition states that the addition of mass exitnctions preserves the CPP property of the reconstructed tree and displays a characterization of its coalescent distribution. It is proved in the Appendix.

\begin{prop}
\label{prop:bottlenecks}
Start with a CPP tree with inverse tail distribution $F$. Add extra mass extinctions with survival probabilities $\varepsilon_1,\ldots,\varepsilon_k$ at times $\timestop-s_1>\ldots>\timestop-s_k$ (where $s_1>0$ and $s_k<\timestop$). Then conditional on survival, the reconstructed tree of the phylogenetic tree obtained after the passage of mass extinctions is again a coalescent point process with inverse tail distribution $F_\varepsilon$ given by
\begin{equation}
\label{eqn:bottlenecks}
F_\varepsilon (t)=
\varepsilon_1\cdots\varepsilon_m\,F(t)+\sum_{j=1}^m (1-\varepsilon_j)\,\varepsilon_1\cdots\varepsilon_{j-1}\,F(s_j)\qquad t\in[s_m, s_{m+1}], m\in\{0,1,\ldots, k\},
\end{equation}
where $s_0:=0$ and $s_{k+1}=\timestop$ (empty sum is zero, empty product is 1).
\end{prop}

This formula can also include sampling by adding a bottleneck with $s_0 =0$ and $\varepsilon_0=p$, resulting in 
$$
F_\varepsilon (t)=
\varepsilon_0\cdots\varepsilon_m\,F(t)+\sum_{j=0}^m (1-\varepsilon_j)\,\varepsilon_0\cdots\varepsilon_{j-1}\,F(s_j)\qquad t\in[s_m, s_{m+1}], m\in\{0,1,\ldots, k\},
$$
which boils down to $F_\varepsilon=1-\varepsilon_0+ \varepsilon_0 F$ when $k=0$ (since $F(0)=1$), as expected from Proposition \ref{prop:sampling}.

\section*{\textsc{What about Kingman coalescent?}}


Another popular way of randomly constructing rooted, binary trees with edge lengths is the Kingman coalescent \citep{Kin82}. For a tree with a finite number $n$ of tips, the model can be described as follows. Start with $n$ labelled lineages and let time run backwards, from tips to root. In the first step, after a random exponential duration with parameter $n(n-1)/2$, one pair of lineages is merged, uniformly chosen among the $n(n-1)/2$ unordered pairs of labelled lineages. This procedure is repeated recursively until the $(n-1)$-st step where the two remaining lineages are merged, after an exponential duration with parameter $1$. Then the probability of any labelled ranked tree shape under this model is 
$$
\prod_{k=2}^n \frac{2}{k(k-1)} = \frac{2^{n-1}}{n!(n-1)!} 
$$
i.e., the Kingman coalescent tree shape is URT (uniform on ranked labelled trees). 

One could wonder if the Kingman coalescent can be built via the CPP procedure. First observe that the node depths of the Kingman coalescent can be arbitrarily large, so it is impossible to equate its law with that of a CPP with stem age $T$, whose node depths are all smaller than $T$. But then to go round this problem, we could try to set $T=+\infty$ or to randomize $T$, in order to allow for node depths of arbitrary length in the CPP. In other words, for each fixed $n$ we ask the following question $Q_n$:
    \begin{myindentpar}{1cm}
 Question $Q_n$: ``Are there random variables $T_n$ and $A_n$ on $(0,+\infty]$ such that the node depths of the $n$-Kingman coalescent can be obtained by first drawing a realization $T$ of $T_n$ and then $n-1$ independent copies of $A_n$ conditioned to be smaller than $T$ ?''
    \end{myindentpar}
 As usual, we denote by $H_1, H_2,\ldots, H_{n-1}$ the node depths of the CPP with $n$ tips.
Mathematically, we ask if $P_n^{\text{cpp}} = P_n^{\tt K}$, where  $P_n^{\tt K}$ is the law of the Kingman coalescent with $n$ tips and $P_n^{\text{cpp}}$ is the law of the randomized CPP conditioned on $n$ tips, that is
$$
P_n^{\text{cpp}}(H_1 \in dx_1,\ldots, H_{n-1}\in dx_{n-1}) = \int_{(0,+\infty]} \PP(T_n \in dT) \,\prod_{i=1}^{n-1} \PP(A_n \in dx_i\mid A_n <T) 
$$
The distribution $P_n^{\text{cpp}}$ has been considered in Aldous \& Popovic (2005) \citep{AlPo2005}, in the special case when 
\begin{itemize}
\item
$A_n$ has the law of node depths in the reconstructed tree of the critical birth--death process with (constant) rates (both) equal to $b_n$;
\item
$T_n$ is given the (improper) density equal to 1 everywhere, further conditioned on this  birth--death process (started at 0 and stopped at $T_n$) to have $n$ tips (which makes it a proper random variable).
\end{itemize}
In other words, here $P_n^{\text{cpp}}$ is the law of the reconstructed tree of a critical birth--death process with a `uniform' prior on its stem age, further conditioned to have $n$ tips.
In particular in this case, $\PP(T_n\in dT)/dT = n\,b_n^n\,T^{n-1}/ (1+b_nT)^{n+1}$ and $\PP(A_n>T) =  1/(1+b_nT)$. The second author of the present paper has shown in \citep{Gernhard2008BMB} that when $b_n = n/2$, the \emph{expectations} of node depths under $P_n^{\tt K}$ are equal to the \emph{expectations} of node depths under $P_n^{\text{cpp}}$. However, she also proved that these two probabilities are \emph{not} equal, so that the answer to $Q_n$ is `no' for this special case of randomized CPP, despite the equality between expectations of node depths.

\begin{prop} The answer to $Q_2$ is `yes', but the answer to $Q_n$ cannot be `yes' for infinitely many $n$.
\end{prop}
The answer to $Q_2$ is  `yes' since for $n=2$ we can set $T_2=+\infty$ and choose $A_2$ as the exponential random variable with parameter 1. Now assume that $P_n^{\text{cpp}}=P_n^{\tt K}$ for infinitely many $n$.
First, observe that under $P_n^{\text{cpp}}$ the node depths are exchangeable, that is, their law is invariant under any permutation. Second, recall that the laws $P_n^{\tt K}$ converge as $n\to\infty$ to the law $P_\infty^{\tt K}$ of what is known as the standard Kingman coalescent (i.e., starting at infinity). It is known that under $P_\infty^{\tt K}$, the node depths can be ranked in decreasing order in a single sequence converging to 0. Then since $P_n^{\text{cpp}}=P_n^{\tt K}$ for infinitely many $n$, there is a subsequence of $(P_n^{\text{cpp}})$ converging to $P_\infty^{\tt K}$. This forces the sequence of node depths under $P_\infty^{\tt K}$ to be exchangeable. This yields a contradiction, since by de Finetti's representation theorem of infinite sequences of exchangeable random variables, no such sequence can be ranked in decreasing order.

Our question remains unsolved but we conjecture that the answer to $Q_n$ is `no' for all $n\ge 3$.

\section*{\textsc{Discussion}}

In this paper, we characterized the forward-time macroevolutionary models which have URT reconstructed trees, and among those which have CPP reconstructed trees. We showed that reconstructed trees are CPP if speciation and extinction rates may only depend on time, and extinction rates may additionally depend on a non-heritable trait, in the case of asymmetric speciation. For all these models, reconstructed tree shapes follow the URT distribution. 

When the speciation or extinction rate depends additionally on the number of species, reconstructed trees are not CPP any longer, however, their ranked tree distribution again is the URT distribution.
For all remaining model classes, we have displayed examples where the ranked tree distribution is not the URT distribution.

We end the paper outlining how to use the results in empirical studies. Phylogenies with branch lengths are increasingly becoming available from empirical data, and such phylogenies have been fitted to speciation and extinction models in order to quantify speciation and extinction rates \citep{Stadler2011PNAScommentrary}. However, the inference methods had to make restrictive assumptions on the speciation and extinction rates. Here we provide a general framework for calculating the likelihood of a reconstructed phylogeny for any model under which the speciation and/or extinction rate may depend on time and the extinction rate may additionally  depend on an asymmetric non-heritable trait (model class 1). The likelihood calculation is based on a CPP representation and allows for model selection and the quantification of the model parameters, using e.g. maximum likelihood or Bayesian methods.
In order to do the statistical analysis of the empirical trees using our CPP representation,  the following numerical steps have to be performed:

\begin{enumerate}
\item Compute the density of species lifetimes $g$ by Equations \eqref{eqn:hts}, \eqref{eqn:ustx}, and \eqref{eqn:Feynman-Kac}, or very simply \eqref{eqn:FK age} in the case when the trait is the age, even if there also is time-dependence;
\item Compute $(F,F')$ by solving the integro-differential  equation \eqref{eqn:hidden convolution};
\item For adding bottlenecks or sampling, compute $F_p$ or $F_\varepsilon$ and their derivatives using \eqref{eqn:bottlenecks};
\item Compute the likelihood of  the tree using \eqref{eqn:likelihood1}, \eqref{eqn:likelihood1bis} or \eqref{eqn:likelihood2} (or in the case of incomplete sampling using the equations in Section ``Missing tips'').
\item Use likelihood in maximum likelihood  or Bayesian  parameter inference.
\end{enumerate}

If we want to simulate trees, we use the coalescent distribution $F$, $F_p$ or $F_\varepsilon$ and sample the speciation times from this one-dimensional distribution. Analytic solutions for $F$ are only known in the case of constant rates, piecewise constant rates, or a Gamma-distributed lifetime with shape 2 (see above); semi-analytic solutions are known for deterministic lifetimes. For other models, numerical approaches become necessary, and the main challenge of a future study \citep{AlexanderEtAl2013} will be to address point 2 with satisfying accuracy and efficiency.

\paragraph{Acknowledgments.} AL was financially supported by grant MANEGE `Mod\`eles Al\'eatoires en \'Ecologie, G\'en\'etique et \'Evolution' 09-BLAN-0215 of ANR (French national research agency). AL also thanks the {\em Center for Interdisciplinary Research in Biology} (Collège de France) for funding. TS thanks the Swiss National Science foundation for funding (SNF grant \#PZ00P3 136820). 
The authors thank Rampal S. Etienne for proposing the use of Lemma \ref{lem} (Appendix). They thank Mike Steel and Helen Alexander for their careful reading and relevant comments.

 \bibliographystyle{sysbio}

\newpage
\appendix

\section*{\textsc{Proofs of technical results}}

\subsection*{\textit{Proof of Proposition \eqref{prop:n-sampling} and Equations \eqref{eqn:n-sampling} and \eqref{eqn:n-sampling1}}}

Recall that $H_i'$ denotes the coalescence time between sampled tip $i$ and sampled tip $i+1$, $i=1,\ldots, n-1$. 
Observe that each sampling configuration $\vec{m}=(m_0,\ldots,m_n)$ such that $m_0+\cdots+m_n=m$ has the same probability, which can be interpreted as a single way of choosing $n$ among $n+m$ labelled balls, so that 
$$
P(\vec{m}) = \frac{n! \ m!}{(n+m)!} .
$$
By summing over all possible sampling configurations, the same argument as in the paragraph on Bernoulli sampling implies that for any $m\ge 0$ and any $x_1, \ldots, x_{n-1}\in [0,T]$
\begin{multline}
\label{eqn:big sum}
P(N_T=n+m, H_1' <x_1, \ldots, H_{n-1}'<x_{n-1}) = \frac{n! \ m!}{(n+m)!} \ P(H>T)\ \times\\ \times\ \sum_{\vec{m}:m_0+\cdots+m_n=m} P(H<x_1)^{m_1+1}\cdots P(H<x_{n-1})^{m_{n-1}+1} P(H<T)^{m_0+m_n}.
\end{multline}


It is easy to differentiate \eqref{eqn:big sum} to get 
\begin{multline*}
P(N_T=n+m, H_1' \in dx_1, \ldots, H_{n-1}'\in dx_{n-1})/dx_1\cdots dx_{n-1}\\ = \frac{n! \ m!}{(n+m)!} P(H>T) \sum_{\vec{m}:m_0+\cdots+m_n=m}  P(H<T)^{m_0+m_n}\prod_{i=1}^{n-1}(m_i+1)f(x_i)P(H<x_i)^{m_i}  .
\end{multline*}
If we sum directly over all pairs $(m_0,m_n)$, and if we write $x_n= T$, we get
\begin{multline*}
P(N_T=n+m, H_1' \in dx_1, \ldots, H_{n-1}'\in dx_{n-1})/dx_1\cdots dx_{n-1}\\ = \frac{n! \ m!}{(n+m)!} P(H>T) \sum_{\vec{m}:m_1+\cdots+m_n=m}  (m_n+1)P(H<x_n)^{m_n}\prod_{i=1}^{n-1}(m_i+1)f(x_i)P(H<x_i)^{m_i}  .
\end{multline*}
This proves \eqref{eqn:likelihoodmissing}. A similar line of reasoning shows that the same correction factor holds for the reconstructed tree with crown age $T$.

Let us prove Equation \eqref{eqn:n-sampling}. With the new notation, we can rewrite equation \eqref{eqn:big sum} as
$$
P(N_T=n+m, H_1' <x_1, \ldots, H_{n-1}'<x_{n-1}) = \frac{n! \ m!}{(n+m)!}\ (1-p_0)p_1\cdots p_{n-1}\ f_{n+1,m}(p_0,p_1\ldots, p_{n-1},p_0),
$$
where
$$
f_{n,m} (p_1, \ldots, p_n) :=\sum_{(m_1, \ldots, m_n)\in N_m^n}\prod_{i=1}^n p_i^{m_i},
$$
and $N_m^n$ is the set of $n$-tuples of integers $(m_1, \ldots, m_n)$  such that $\sum_{i=1}^n m_i = m$. Let us state a useful lemma, which is proved at the end of this section.
\begin{lem}
\label{lem}
For all integers $n\ge1$ and $m\ge 0$, for any pairwise distinct $p_1,\ldots, p_n\in[0,1]$, 
$$
f_{n,m} (p_1, \ldots, p_n) = \sum_{i=1}^n \frac{p_i^{m+n-1}}{\prod_{j=1,\ldots, n, j\not=i} (p_i-p_j)}.
$$
\end{lem}
In view of this lemma, we get
$$
P(N_T=n+m, H_1' <x_1, \ldots, H_{n-1}'<x_{n-1}) = \frac{n! \ m!}{(n+m)!}\ (1-p_0)p_1\cdots p_{n-1}\ \lim_{p_n\to p_0} \sum_{i=0}^{n} \frac{p_i^{m+n}}{\prod_{j=0,\ldots, n, j\not=i} (p_i-p_j)}.
$$
Now the following limit holds for any $k=1,\ldots, n$
\begin{multline*}
\lim_{p_i\to p_0, \forall i = n-k+1,\ldots,n}\sum_{i=0}^{n} \frac{p_i^{m+n}}{\prod_{j=0,\ldots, n, j\not=i} (p_i-p_j)}\\
= \sum_{i=1}^{n-k} \frac{p_i^{m+n}}{(p_i-p_0)^{k+1}\prod_{j=1,\ldots, n-k, j\not=i} (p_i-p_j)}
+ \frac{1}{k!}\ \left(\frac{b}{a_{n-k}}\right)^{(k)}(p_0),
\end{multline*}
where 
$$
b(x) := x^{m+n}\quad\mbox{ and }\quad a_k(x) := \prod_{j=1}^{k}(x-p_j).
$$
Applying this to $k=1$, we get \eqref{eqn:n-sampling}.
Actually, it is also possible to apply this to $k=n$ to recover the law of $N$, and to $k=n-1$ to get the law of $H_1'$, namely
\begin{eqnarray*}
P(N_T=n+m, H_1'<x_1) &=& \frac{n! \ m!}{(n+m)!}\ (1-p_0)p_0^{n-2}p_1(p_1-p_0)^{-n}
\\
 &\times&\left(
b(p_1) - b(p_0) - (p_1-p_0) b'(p_0) - \cdots- \frac{(p_1-p_0)^{n-1}}{(n-1)!}\ b^{(n-1)}(p_0)
\right)\\
	&=& n(1-p_0)p_0^{n-2}p_1(p_0-p_1)^{-n} \int_{p_1}^{p_0} y^m (y-p_1)^{n-1}\, dy,
\end{eqnarray*}
which is Equation \eqref{eqn:n-sampling1}.
 
\begin{proof}[Proof of Lemma \ref{lem}]
We will need the following preliminary result. For any $n\ge 1$ and any pairwise distinct real numbers $p_1,\ldots, p_n$, the polynomial $R$ defined by
$$
R(x):= \sum_{i=1}^n p_i^{n-1} \prod_{j=1,\ldots, n, j\not=i} \frac{x-p_j}{p_i-p_j}
$$
is actually also given by $R(x) = x^{n-1}$. This is merely due to the fact that $R$ has degree $n-1$, and hence is characterized by the values it takes at $n$ distinct points (here, the points $p_1,\ldots, p_n$). \\

It is easy to check that the formula holds for $n=1$ and any $m\ge 0$ since $f_{1,m}(s) = s^m$. Let us prove the formula by induction on $n$. Let $n\ge 1$ and assume the formula holds for this integer $n$ and any $m\ge 0$. 
Thanks to this assumption, for any pairwise distinct real numbers $p_1,\ldots, p_{n+1}$,
\begin{eqnarray*}
f_{n+1,m} (p_1,\ldots, p_{n+1}) &=& \sum_{m_{n+1}=0}^m p_{n+1}^{m+1} \sum_{(m_1,\ldots, m_n)\in N_{m-m_{n+1}}^{n}}\prod_{i=1}^n {p_i}^{m_i}\\
	&=& \sum_{m_{n+1}=0}^m p_{n+1}^{m_{n+1}} f_{n,m-m_{n+1}}(p_1, \ldots, p_n)\\
	&=& \sum_{m_{n+1}=0}^m p_{n+1}^{m_{n+1}} \ \sum_{i=1}^n \frac{p_i^{m-m_{n+1}+n-1}}{\prod_{j=1,\ldots, n, j\not=i} (p_i-p_j)} \\
		&=& \sum_{i=1}^n \frac{p_i^{n-1}}{\prod_{j=1,\ldots, n, j\not=i} (p_i-p_j)} \ \frac{p_i^{m+1}-p_{n+1}^{m+1}}{p_i- p_{n+1}}\\
			&=& \sum_{i=1}^n \frac{p_i^{m+n}}{\prod_{j=1,\ldots, n+1, j\not=i} (p_i-p_j)} +  p_{n+1}^{m+1} \frac{R(p_{n+1})}{\prod_{j=1}^n (p_{n+1}-p_j)} ,
\end{eqnarray*}
where $R$ is the polynomial introduced in the beginning of this proof. Since we have shown that $R(x) = x^{n-1}$, the result is proved by the induction principle. 

 \end{proof}

\subsection*{\textit{Proof of Theorem \ref{ThmMain}}}

Let $n\ge 1$ be an integer, $h_0=T$ and $h_1,\ldots, h_{n-1}$ be elements of $(0,T)$. Assume $N_T\ge n$, and condition on $H_i = h_i$ for $i=0,\ldots, n-1$. We are going to prove that the conditional law of $H_n$ is given by 
\begin{equation}
\label{to prove}
P(H_n>t) = \exp\left(- \int_{\timestop-t}^ \timestop \lambda(s)\, (1-q(s))\,ds\right)\qquad t\in[0,T],
\end{equation}
which will show that $H_n$ is independent of $H_1, \ldots, H_{n-1}$ and has $F$ as inverse tail distribution. This result yields the theorem by induction. Note that conditonal on $N_T\ge n$, $N_T$ exactly equals $n$ iff $H_n >T$.

Label by $0,1,\ldots, n-1$ the extant species  at time $\timestop$ in the order induced by the orientation of the tree. In particular, $h_i$ is the coalescence time between species $i-1$ and $i$ ($0\le i\le n-1$).

We denote by $k$ the number of generations separating species $n-1$ from the progenitor species. We let $x_k$ denote the time when species $n-1$ was born, $x_{k-1}<x_k$ the time when her mother was born, and so on, until 
$x_0=0$ the birth time of the progenitor species. By the orientation of the tree, there are (conditionally) deterministic indices $0=i_0  <\cdots <i_k=n-1$, such that $x_j= T-h_{i_j}$ (and $h_v< h_{i_j}$ for all $v\in\{i_{j-1}+1,\ldots, i_j -1\}$), so that conditional on $H_i = h_i$ for $i=0,\ldots, n-1$, the times $x_0,\ldots, x_k$ are deterministic. 

By the orientation of the tree again, apart from the species already labelled, species extant at $T$ descend from speciations occurring during one of the time intervals $I_j:=[x_j, x_{j+1})$, where $x_{k+1}:=T$ for convenience. On each of these time intervals, speciations occur at rate $\lambda(t)$, and so successful speciations, i.e., speciations with extant descendance at time $T$, occur at rate $\lambda(t)\, (1-q(t))$. But conditional on the $(x_j)$, all the ancestors of species $n-1$ (including her) independently speciate on their corresponding interval $I_j$. Then if $A$ denotes any subset of $[0,T)$, the number $N(A)$ of successful speciations occurring during $A$ is the sum 
$$
N(A) = \sum_{j=0}^k N(A\cap I_j),
$$
where the random numbers $N(A\cap I_j)$ are independent. Now from Section ``Precisions about rates'', $N(A\cap I_j)$ is a Poisson random variable with parameter $\int_{A\cap I_j} \lambda(t)\, (1-q(t))\, dt$. As a consequence, $N(A)$ is a Poisson random variable with parameter
$$
\sum_{j=0}^k \int_{A\cap I_j} \lambda(t)\, (1-q(t))\, dt = \int_{A} \lambda(t)\, (1-q(t))\, dt.
$$
The proof finishes noticing that $H_n>t$ iff $N([T-t, T)) =0$, which occurs with the probability displayed in \eqref{to prove}.

\subsection*{\textit{Proof of Proposition \ref{prop:NEW}}}

Invoking arguments from Section ``Precisions about rates'', it is easy to see that 
$$
q(t) = \int_t^T ds\, g(t,s) \,e^{-\int_t^s du \lambda(u)\,(1-q(u))}.
$$
Recalling that 
$$
F(t) = \frac{1}{P(H>t)} = \exp\left( \int_{T-t}^T \lambda(s)\, (1-q(s))\,ds\right),
$$
we get
$$
q(t)= \int_t^T ds\, g(t,s) \, \frac{F(T-s)}{F(T-t)},
$$
or equivalently
$$
q(T-t) = \int_0^t ds\, g(T-t,T-s) \, \frac{F(s)}{F(t)}.
$$
Now check that 
$$
F'(t) = \lambda(T-t)\, (1-q(T-t))\, F(t).
$$
Equation \eqref{eqn:hidden convolution} is a consequence of the last two displayed equations.

\subsection*{\textit{Proof of Proposition \ref{prop:markov}}}

By an integration by parts on Equation \eqref{eqn:hidden convolution} in Proposition \ref{prop:NEW},  we get
$$
F'(t) = \tilde{\lambda} (t) \left[e^{-\int_0^t \tilde{\mu}(u)\,du} + \int_0^t ds\, F'(s)\,e^{-\int_s^t \tilde{\mu}(u)\,du}\right],
$$
where
$$
\tilde{\lambda}(t) :=\lambda (T-t) \quad \mbox{ and }\quad \tilde{\mu}(t):= \mu (T-t).
$$
Setting 
$$
G(t) := F'(t)\,e^{\int_0^t \tilde{\mu}(u)\,du},
$$
we get
$$
G(t) = \tilde{\lambda} (t) \left(1 + \int_0^t ds\, G(s)\right),
$$
which yields 
$$
G(t)= \tilde{\lambda} (t) \,e^{\int_0^t \tilde{\lambda}(u)\,du}.
$$
Since $F(0)=1$, we have 
$$
F(t) = 1+ \int_0^t F'(s) \, ds = 1+ \int_0^t G(s)\,e^{-\int_0^s \tilde{\mu}(u)\,du}
	= 1+ \int_0^t \tilde{\lambda} (s) \,e^{\int_0^s (\tilde{\lambda}- \tilde{\mu})(u)\,du},
$$
which proves the result after changing variables.

\subsection*{\textit{Proof of Proposition \ref{PropnMh}}}

Since $g(t,s)= g(t-s)$, \eqref{eqn:hidden convolution} becomes \eqref{eqn:convolution}, that is
$$
F'(t) = \lambda\,\left( F(t) - \int_0^t F(s) g(t-s)\, ds\right)= \lambda\,\left( F(t) - F\star g(t)\right),
$$
where $\star$ denotes the convolution product. Denoting by $\mathscr{L}$ the Laplace transform and using the fact that $\mathscr{L}(F\star g)=\mathscr{L}(F)\mathscr{L}(g)$, we get
$$
\mathscr{L}(F') = \lambda \mathscr{L}(F)(1- \mathscr{L}(g)).   
$$
An integration by parts shows that 
$$
\mathscr{L}(F')(x) = \int_0^\infty F'(t) \,e^{-tx}\, dt = \left[F(t)\,e^{-tx}\right]_0^\infty + x\int_0^\infty F(t) \,e^{-tx}\, dt,
$$
so that as soon as $x$ is larger than the exponential growth parameter of $F$ (which is the largest root of the convex function $\psi$), we have $\mathscr{L}(F')(x) = -1 + x \mathscr{L}(F)(x)$. As a consequence,
$$
(x-\lambda + \lambda \mathscr{L}(g)(x))\mathscr{L}(F)(x) = 1, 
$$
which is the announced equality.

\subsection*{\textit{Proof of Proposition \ref{prop:bottlenecks}}}

We first characterize the effect of one bottleneck on a coalescent point process and then generalize to $k$ bottlenecks by iterative thinnings.

Assume $k=1$ and $s_1\in(0,\timestop)$. Recall that a coalescent point process is defined thanks to a sequence of independent, identically distributed random variables $(H_i)$. We will see that the tree obtained after thinning is still a coalescent point process, defined from independent random variables, say $(B_i)$. Let $(e_i)$ be the i.i.d. Bernoulli random variables defined by $e_i=1$ if lineage $i$ survives the bottleneck (this has a meaning only if $H_i \ge s_1$; it happens with probability $\varepsilon_1$). By the orientation of the tree, a tip terminating a branch with  depth smaller than $s_1$ is kept alive iff the last branch with depth larger than $s_1$ is not thinned at time $\timestop-s_1$. As a consequence, if $H_1<s_1$, then the first lineage is alive and its coalescence time with the left-hand ancestral lineage is $B=H_1$. Otherwise, define $1=J_1< J_2<\cdots$ the indices of consecutive branches with depths larger than $s_1$. Then the first lineage kept alive after thinning is 
the least $J_m$ such that $e_{J_m}=1$, and its coalescence time with the ancestral lineage is $B=\max(H_{J_1},\ldots, H_{J_m})$. By the independence property of coalescent point processes and by the independence of the Bernoulli random variables $(e_i)$, the new genealogy is obtained by a sequence of independent random variables $(B_i)$ all distributed as $B$. 

Let us specify a little bit the law of $B$. First, with probability $P(H< s_1)$, $P(B\in \cdot) =P(H\in \cdot \mid H< s_1)$. Second, with probability $P(H\ge s_1)$
$$
B\stackrel{(d)}{=}\max\{A_1,\ldots, A_K\},
$$
where the $A_i$'s are i.i.d. distributed as $H$ conditional on $H\ge s_1$ and $K$ is an independent (modified) geometric r.v., that is, $P(K=j)=\varepsilon_1(1-\varepsilon_1)^{j-1}$. Similarly as in Proposition \ref{prop:sampling}, we get for any $s\ge s_1$
$$
\frac{1}{P(B\ge s)}=\frac{1-\varepsilon_1}{P(H \ge s_1)} + \frac{\varepsilon_1}{P(H\ge s)}\qquad s\ge s_1. 
$$
Then if $F_\varepsilon$ denotes the inverse tail distribution of $B$, i.e., $F_\varepsilon(s):= 1/P(B\ge s)$, we have
$$
F_\varepsilon (s) = 
\begin{cases}
F(s) &\text{if } 0\le s\le s_1\\
(1-\varepsilon_1)F(s_1)+\varepsilon_1 F(s)&\text{if }s_1\le s \le t, 
\end{cases}
$$
where $F$ is the inverse tail distribution of $H$.
Iterating this procedure yields the result in Proposition \ref{prop:bottlenecks}.

\section*{\textsc{Generator and Feynman-Kac formulae}}

We first define what is the generator $L_t$ of the trait dynamics. If the trait takes values in a finite (or even countable) state-space $E$, then for any function $\varphi:E\to \RR$, the function $L_t\varphi:E\to \RR$ is defined by 
$$
L_t \varphi(x) = \sum_{y\in E, y\not=x} \rho_t(x,y)\ (\varphi(y)-\varphi(x)), 
$$
where $\rho_t(x,y)$ is the jump rate from trait value $x$ to trait value $y$, at time $t$. If $X$ is a diffusion process on the real line, and if $\varphi$ is twice differentiable, then 
$$
L_t\varphi(x) = c(t,x)\ \varphi'(x) + \frac{1}{2} \sigma^2(t,x)\ \varphi''(x),
$$
where $c(t,x)$ is the infinitesimal mean and $\sigma^2(t,x)$ is the infinitesimal variance. This includes the case when $X$ is driven by a differential equation ($\sigma\equiv 0$), which in turn includes the case when $X$ is the age ($c\equiv 1$, since then $dX_t = dt$).

We now explain why we did not follow the path of characterizing $q(t)$ in terms of the model parameters. Let $q(t,x)$ denote the probability that a species carrying trait value $x$ at time $t$ has no descendance by time $\timestop$. In particular,
$$
q(t) = \int_{\RR} \nu_t(dx)\, q(t,x),
$$
and in the case when the trait is the age, $q(t) = q(t,0)$ ($\nu_t$ is then the Dirac measure at 0).

If the process $X$ is Markovian, then it is possible to prove that 
\begin{equation}
\label{eqn:generator}
\frac{\partial q}{\partial t} (t,x) + L_t q (t,x) 
= -\lambda(t)\ q(t,x)\ q(t) - \mu(t,x) +(\lambda(t) + \mu(t,x))\ q(t,x),
\end{equation}
with the terminal condition $q(T,x)=0$. Solving \eqref{eqn:generator} will yield an expression of $q(t,x)$ involving $q(t)$ as an argument. Integrating $q(t,x)$ against $\nu_t(dx)$ will yield $q(t)$, which will then appear on both sides of the equation. Since it is not always clear in general how to identify $q(t)$ from this equation, we have provided an alternative solution in Proposition \ref{prop:NEW}.

\section*{\textsc{Beyond rates}}

We end this appendix by highlighting that all the results in this paper would still hold even if speciation and extinction rates were not proper rates. For example, in Section ``Deterministic lifetimes'', we have applied Proposition \ref{PropnMh} to the case when species lifetimes can be deterministic, fixed to some value, say $b$. This amounted to replacing the lifetime distribution $g(a)\,da$ by the Dirac measure $\delta_b(da)$. 

Here, the number $N(a,b)$ of speciations from a same mother species between times $a$ and $b$ is a Poisson random variable with parameter $\int_a^b \lambda(t)\, dt$. It is usual, as we do here, to say that speciations occur at rate $\lambda(t)$ at time $t$, but we could also say that speciation dates are the `atoms' of a `Poisson point measure' with intensity measure $\lambda(t)\, dt$. This definition can be extended to considering a Poisson point measure with intensity measure $\Lambda$, where $\Lambda$ is any finite measure on the real numbers. This means that $N(a,b)$ now is a Poisson random variable with parameter $\Lambda ([a,b])=\int_a^b \Lambda(dt)$. For example, if $\Lambda=\sum_{i=1}^n \alpha_i \delta_{s_i}$, with $\alpha_i>0$ for all $i$ and $0<s_1<\cdots <s_n<\timestop$, then $N(a,b)$ is obtained by adding independent numbers of speciation events at each time $s_i$, each following the Poisson distribution with parameter $\alpha_i$. All statements hold, including Theorem \ref{ThmMain}, if we replace everywhere  $\lambda(t)\, dt$ with a general $\Lambda(dt)$ and $\mu(t) \,dt$ by a general $M(dt)$. 

In particular, in the case when rates only depend on time, Proposition \ref{prop:bottlenecks} (CPP with mass extinctions) can be seen as a generalized version of Proposition \ref{prop:markov} where we add to the extinction intensity measure $M_0(dt) =\mu(t)\, dt$ a point measure $M_1$ with atoms at the time points where mass extinctions occur. More specifically, with 
$$
M_1=\sum_{i=1}^k \ln(1/\varepsilon_i) \,\delta_{T-s_i},
$$
the reconstructed tree of the diversification model with extinction intensity measure $M_0+M_1$ is the reconstructed tree of the diversification model with extinction intensity measure $M_0$ (i.e., with extinction rate $\mu(t)$ at time $t$) to which mass extinctions are added at the times $T-s_i$ with survival probability $\varepsilon_i$. Indeed, there will be an additional extinction at time $T-s_i$ if the number of atoms at $T-s_i$ of the Poisson point measure of extinctions is nonzero, which happens with probability $1-\exp(-\ln(1/\varepsilon_i) ) = 1-\varepsilon_i$. Now recall from Proposition \ref{prop:markov} that in the Markovian case,
$$
F(t) =1+\int_{T-t}^T ds\,\lambda(s)\,e^{\int_s^T dr\,(\lambda(r)-\mu(r))}.
$$
It can be seen that if we replace in the previous equation the measure $\mu(r)\, dr$ by $\mu(r)\, dr + M_1(dr)$, we indeed recover the function $F_\varepsilon$ displayed in Proposition \ref{prop:bottlenecks}.

\end{document}